\documentclass[11pt]{amsart}
\usepackage{amsmath,amssymb,graphicx}
\usepackage[english]{babel}
\usepackage[T1]{fontenc}
\usepackage{amscd,amsthm, amsfonts,}
\usepackage{hhline, latexsym}
\usepackage{float}

\newlength{\dinwidth}
\newlength{\dinmargin}
\newcommand{\R}{\mathbb{R}}
\newcommand{\N}{\mathbb{N}}
\newcommand{\C}{\mathbb{C}}
\newcommand{\Z}{\mathbb{Z}}
\newcommand{\B}{\mathbb{B}}
\newcommand{\z}{\mathbf{z}}
\newcommand{\mr}{\mathbf{r}}

\newcommand{\ud}{\,\mathrm{d}}
\newcommand{\Rs}{\mathcal{R}}

\newtheorem{theorem}{Theorem}[section]
\newtheorem{proposition}{Proposition}[section]

\begin{document}

\def\theequation {\thesection.\arabic{equation}}
\makeatletter\@addtoreset {equation}{section}\makeatother

\title[Algebro-geometric solutions to the Camassa-Holm equation]{New construction of algebro-geometric solutions to the Camassa-Holm equation and their numerical evaluation } 

\author{C.~Kalla}
\address{Institut de Math\'ematiques de Bourgogne,
		Universit\'e de Bourgogne, 9 avenue Alain Savary, 21078 Dijon
		Cedex, France}
    \email{Caroline.Kalla@u-bourgogne.fr}

\author{C.~Klein}
\address{Institut de Math\'ematiques de Bourgogne,
		Universit\'e de Bourgogne, 9 avenue Alain Savary, 21078 Dijon
		Cedex, France}
    \email{Christian.Klein@u-bourgogne.fr}

\date{\today}    

\begin{abstract}
    An independent derivation of solutions to the Camassa-Holm equation in terms of multi-dimensional theta functions is
    presented using an approach based on Fay's identities.  Reality and 
    smoothness conditions are studied for these solutions from the 
    point of view of the topology of the underlying real 
    hyperelliptic surface. The solutions are studied numerically for 
    concrete examples, also in the limit where the surface 
    degenerates to the Riemann sphere, and where solitons and cuspons 
    appear.
\end{abstract}


\thanks{We thank V.~Shramchenko for helpful discussions and for 
carefully reading the manuscript.
This work has been supported in part by the project FroM-PDE funded by the European
Research Council through the Advanced Investigator Grant Scheme, the Conseil R\'egional de Bourgogne
via a FABER grant, the Marie-Curie IRSES program RIMMP and the ANR via the program ANR-09-BLAN-0117-01. }

\maketitle

\section{Introduction}

The Camassa-Holm (CH) equation 
\begin{equation}
u_{t}+3\,uu_{x}=u_{xxt}+2\,u_{x}u_{xx}+uu_{xxx}-2k\,u_{x}, \label{CH intro}
\end{equation}
was first found by
Fokas and Fuchssteiner \cite{FF} with the method of
recursion operators and shown to be a bi-hamiltonian equation with
an infinite number of conserved functionals.  Camassa and Holm  
\cite{CH} showed that it appeared as a model for
unidirectional propagation of waves in shallow water, $u(x,t)$
representing the height of the free surface about a flat bottom,
$k$  being a constant related to the critical shallow water
speed. In this context, only real-valued solutions are physically meaningful.

To be able to formulate a Cauchy problem for the CH equation which 
implies the solution of 
(\ref{CH intro}) for a given function $u(x,0)$, it is convenient to 
write (\ref{CH intro}) in the non-local evolutionary form,
\begin{equation}\label{CHevo}
    u_{t}=D^{-1}(-3uu_{x}+2\,u_{x}u_{xx}+uu_{xxx}-2k\,u_{x}),
\end{equation}
where the operator $D$ is given by $D=1-\partial_{xx}$, and where its 
inverse is defined by giving certain boundary conditions. Since the 
inverse of the operator $D$ stands for an integral over the Green's 
function for given boundary conditions, and thus an integral from some 
base point $x_{0}$ to $x$, it is not a local operation in $x$. This 
non-locality has mathematically interesting consequences: the CH 
equation has traveling wave solutions of the form 
$u(x,t)=c\,\exp\{-|x-vt|\}$ ($v$ being the speed, $c=const$) called peakons that have a discontinuous first derivative at the wave peak.
Camassa and Holm \cite{CH} described the dynamics of the peakons in terms of a finite-dimensional completely integrable Hamiltonian system, namely, each peakon solution is associated with a mechanical system of moving particles. The class of mechanical systems of this type was further extended by Calogero and  Fran\c coise in \cite{Cal, CF}. Multi-peakon solutions were studied using  different approaches in a series of papers \cite{BSS1, BSS2, BSS3, Cam}. Periodic solutions of the shallow water equation were discussed in \cite{KC}. 

A further consequence of this non-locality is that  
solutions of the CH equation in terms of multi-dimensional theta 
functions do not depend explicitly on the physical coordinates. Such 
solutions were first given by Alber and Fedorov in \cite{AFY1} by 
solving a generalized Jacobi inversion problem. In contrast to the 
well known cases of Korteweg-de Vries (KdV), nonlinear Schr\"odinger and sine-Gordon
equations, see for instance \cite{BBEIM} and references therein, complex solutions of the CH equation are not
meromorphic functions of $(x, t)$ but have several branches. This is 
due to the presence of an implicit function $y(x,t)$ of the variables 
$x$ and $t$ in the argument of the theta function appearing in the 
solutions. A monodromy effect is thus present in the profile of 
real-valued solutions such as cusps and peakons. This means 
that even bounded solutions to the CH equation can have discontinuous 
or infinite derivatives in contrast to KdV solutions. 
Algebro-geometric solutions of the Camassa-Holm equation and their 
properties are studied in \cite{ACFYHM2, AFY1, AFY2,GH1,GH2,GH3}.

Our goal in this paper
is to give an independent derivation of such solutions based on identities between multi-dimensional theta functions, which naturally arise from Fay's identity \cite{Fay}.
This identity states that, for any points $a,b,c,d$ on a compact 
Riemann surface of genus $g>0$, and for any $\z\in\C^{g}$, there 
exist scalars $\gamma_{1},\gamma_{2},\gamma_{3},$ depending on the points $a,b,c,d$, such that
\begin{equation}
\gamma_{1}\,\Theta(\z+\textstyle\int^{a}_{c})\,\Theta(\z+\textstyle\int^{d}_{b})
+\gamma_{2}\,\Theta(\z+\textstyle\int^{a}_{b})\,\Theta(\z+\textstyle\int^{d}_{c})
=\gamma_{3}\,\Theta(\z)\,\Theta(\z+\textstyle\int^{a}_{c}+\textstyle\int^{d}_{b}), \label{Fay intro}
\end{equation}
where $\Theta$ is the multi-dimensional theta function (\ref{theta}); here and below we use the notation $\int^{b}_{a}$ for the Abel map (\ref{abel}) between a and b. 
While the authors in \cite{ACFYHM2, AFY1, AFY2} used  generalized theta
functions and generalized Jacobians (going back to investigations of Clebsch and
Gordan \cite{CG}), we derive the solutions from the identity 
(\ref{Fay intro}). This fits into 
the program formulated by Mumford \cite{Mum} that all algebro-geometric 
solutions to integrable equations should be obtained from Fay's 
identity and suitable degenerations thereof. Historically this 
approach was only able to reproduce solutions already obtained 
via so-called Baker-Akhiezer functions, generalizations of the 
exponential function to Riemann surfaces. The first example of new solutions 
found via the Fay identity was by one of the authors \cite{Kalla} for 
the multi-component nonlinear Schr\"odinger equations, see also 
\cite{CK}. Both methods have specific advantages: for the 
Baker-Akhiezer approach, solutions to the associated linear system 
the integrability condition of which is the studied equation have to 
be constructed on a Riemann surface for a given singularity 
structure. For the Mumford approach the non-trivial task is the 
finding of a suitable degeneration of the Fay identity for the 
studied equation. Once this is done the identification of certain constants in the solutions 
as well as the study of reality ans smoothness conditions is then in 
general more straight forward than in the Baker-Akhiezer approach. We provide here the 
first example for an integrable equation with nonlocal terms in the 
evolutionary form (\ref{CHevo}) as explained there.
The spectral data for the theta-functional solutions to CH consist of a hyperelliptic curve of the form
 $\mu^{2}=\prod_{j=1}^{2g+2} (\lambda-\lambda_{j})$ with three marked 
 points: two of them are interchanged under the involution 
 $\sigma(\lambda,\mu)=(\lambda,-\mu)$, and the third is a ramification 
 point $(\lambda_{j_{0}},0)$.

Our construction of real valued solutions is based on the description 
of the real and imaginary part of the Jacobian associated to a real hyperelliptic curve (i.e., the branch points $\lambda_{j}$ are real or  pairwise conjugate non-real).
In this way, one gets purely 
transcendental conditions on the parameters (i.e., without reference to a divisor 
defined by the solution of a Jacobi inversion problem), such that the solutions are 
real-valued and smooth. It turns out that continuous real valued 
solutions are either smooth or have an 
infinite number of cusp-type singularities. 

Concrete examples for the resulting solutions are studied numerically by using the code for 
real hyperelliptic surfaces \cite{cam,lmp}. This code uses so-called spectral 
methods to compute periods on the surfaces. It allows also to study 
numerically almost degenerate surfaces where the branch points 
collapse pairwise. In this limit, the theta functions break down to 
elementary functions, and the solutions describe solitons or cusps. 
It is noteworthy that the theta-functional solutions thus contain as 
limiting cases all known solutions to the CH equation. 
The quality of the numerics is ensured by testing the identities 
between theta functions which are used to construct the CH solutions 
in this paper. In addition, the solutions are computed on a 
grid and are numerically differentiated. These independent tests 
ensure that the shown solutions are correct to much better than 
plotting accuracy.

The paper is organized as follows: in section 2 we summarize 
important facts on Riemann surfaces, especially Fay's identities for 
theta functions and results on real surfaces. In section 3 we use 
Fay's identities to rederive theta-functional solutions to the CH 
equation, and give reality and smoothness conditions. In section 4 we 
study numerically concrete examples, also in almost degenerate 
situations. We add some concluding remarks in section 5.

\section{Theta functions and real Riemann surfaces}

In this section we recall basic facts on Riemann surfaces, in particular real surfaces 
and multi-dimensional theta functions defined on them.

\subsection{Theta functions}
Let $\Rs _{g}$ be a 
compact Riemann surface of genus $g>0$. Denote by $(\mathcal{A},\mathcal{B}):= (\mathcal{A}_{1},\ldots,\mathcal{A}_{g},\mathcal{B}_{1},\ldots,\mathcal{B}_{g})$ a canonical homology basis, and by $(\omega_{1},\ldots,\omega_{g})$ the basis of holomorphic differentials normalized via
\begin{equation}
\int_{\mathcal{A}_{k}}\omega_{j}=2\mathrm{i}\pi\delta_{kj}, \quad k,j=1,\ldots,g. \label{norm hol diff}
\end{equation}
The matrix $\B=\left(\int_{\mathcal{B}_{k}}\omega_{j}\right)$ of 
$\mathcal{B}$-periods of the normalized holomorphic differentials $\omega_{j}$, $j=1,\ldots,g$, is symmetric and has a negative definite real part. The theta function with (half integer) characteristics $\delta=[\delta_{1},\delta_{2}]$ is defined by
\begin{equation}
\Theta_{\B}[\delta](\z)=\sum_{\mathbf{m}\in\Z^{g}}\exp\left\{\tfrac{1}{2}\langle \B(\mathbf{m}+\delta_{1}),\mathbf{m}+\delta_{1}\rangle+\langle \mathbf{m}+\delta_{1},\z+2\mathrm{i}\pi\delta_{2}\rangle\right\},\label{theta}
\end{equation}
for any $\z\in\C^{g}$; here $\delta_{1},\delta_{2}\in \left\{0,\frac{1}{2}\right\}^{g}$ 
are the vectors of the characteristics $\delta$; $\langle.,.\rangle$ denotes the 
scalar product $\left\langle \mathbf{u},\mathbf{v} \right\rangle=\sum_{i}u_{i}\,v_{i}$ for any $\mathbf{u},\mathbf{v}\in\C^{g}$. The theta function $\Theta[\delta](\z)$ is even if 
the characteristics $\delta$ is even, i.e., $4\left\langle 
\delta_{1},\delta_{2} \right\rangle$ is even, and odd if the 
characteristics 
$\delta$ is odd i.e., $4\left\langle \delta_{1},\delta_{2} 
\right\rangle$ is odd. An even characteristics is called non-singular if 
$\Theta[\delta](0)\neq 0$, and an odd characteristics is called non-singular if the gradient $\nabla\Theta[\delta](0)$ is non-zero. The theta function with characteristics is related to 
the theta function with zero characteristics (the Riemann theta 
function denoted by $\Theta$) as 
follows
\begin{equation}
\Theta[\delta](\z)=\Theta(\z+2\mathrm{i}\pi\delta_{2}+\B\delta_{1})\,\exp\left\{\tfrac{1}{2}\langle \B\delta_{1},\delta_{1}\rangle+\langle\z+2\mathrm{i}\pi\delta_{2},\delta_{1}\rangle\right\}.\label{2.3}
\end{equation}

Denote by $\Lambda$ the lattice $\Lambda=\{2\mathrm{i}\pi \mathbf{N}+\B 
\mathbf{M}, \,\,\mathbf{N},\mathbf{M}\in\Z^{g}\}$ generated by the 
$\mathcal{A}$ and $\mathcal{B}$-periods of the normalized holomorphic differentials $\omega_{j}$, $j=1,\ldots,g$. The complex torus $J:=J(\Rs_{g})=\C^{g} / \Lambda$ is called the Jacobian of the Riemann surface $\Rs_{g}$. The theta function (\ref{theta}) has the following quasi-periodicity property with respect to the lattice $\Lambda$:
\begin{equation}
\Theta[\delta](\mathbf{z}+2\mathrm{i}\pi \mathbf{N}+\B \mathbf{M}) \nonumber
\end{equation}
\begin{equation}
=\Theta[\delta](\z)  \exp\left\{-\tfrac{1}{2}\langle \B\mathbf{M},\mathbf{M}\rangle-\langle \z,\mathbf{M}\rangle+2\mathrm{i}\pi(\langle\delta_{1},\mathbf{N}\rangle-\langle\delta_{2},\mathbf{M}\rangle)\right\}.\label{per theta}
\end{equation}

Denote by $\Pi$ the Abel map $\Pi:\Rs_{g}\longmapsto J$ defined by 
\begin{equation}
\Pi(p)=\int_{p_{0}}^{p}\omega,  \label{abel}
\end{equation} 
for any $p\in\Rs_{g}$, where $p_{0}\in\Rs_{g}$ is the base point of the application, and where $\omega=(\omega_{1},\ldots,\omega_{g})^{t}$ is the vector of the normalized holomorphic differentials. In the whole paper we use the notation $\int_{a}^{b}=\Pi(b)-\Pi(a)$.

Now let $k_{a}$ denote a local parameter near $a\in\Rs_{g}$ and consider 
the following expansion of the normalized holomorphic differentials $\omega_{j}$, $j=1,\ldots,g$,
\begin{equation} 
\omega_{j}(p)= \left(V_{a,j}+W_{a,j}\,k_{a}(p)+\ldots\right)\,\ud k_{a}(p), \label{exp hol diff}
\end{equation}
for any point $p\in\Rs_{g}$ in a neighborhood of $a$, where $V_{a,j},\,W_{a,j}\in\C$.
Let us denote by $D_{a}$ the operator of directional derivative along the vector $\mathbf{V}_{a}=(V_{a,1},\ldots,V_{a,g})^{t}$:
\begin{equation} 
D_{a}F(\z)=\sum_{j=1}^{g}\partial_{z_{j}}F(\z) V_{a,j}, \label{Da}
\end{equation}
where $F:\C^{g}\longrightarrow \C$ is an arbitrary function.
According to \cite{Mum}, the theta function satisfies the following identities derived from Fay's identity (\ref{Fay intro}):
\begin{align}
D_{b}\ln\frac{\Theta(\z+\int^{a}_{c})}{\Theta(\z)}&=\,p_{1}+p_{2}\,\frac{\Theta(\z+\int^{a}_{b})\,\Theta(\z+\int^{b}_{c})}{\Theta(\z+\int^{a}_{c})\,\Theta(\z)}\,, \label{cor Fay1}\\
D_{a}D_{b}\ln\Theta(\z)&=q_{1}\,+\,q_{2}\,\frac{\Theta(\z+\int^{b}_{a})\,\Theta(\z-\int^{b}_{a})}{\Theta(\z)^{2}}\,,
\label{cor Fay2}
\end{align}
for any $\z\in\C^{g}$ and any distinct points $a,b,c\in\Rs_{g}$; here the scalars $p_{i},q_{i}$, $i=1,2$ depend on the points $a,b,c$ and are given by
\begin{align}
p_{1}(a,b,c)&=-D_{b}\ln\frac{\Theta[\delta](\int^{b}_{a})}{\Theta[\delta](\int^{b}_{c})}\,, \label{p1}\\  p_{2}(a,b,c)&=\frac{\Theta[\delta](\int^{a}_{c})}{\Theta[\delta](\int^{a}_{b})\Theta[\delta](\int^{c}_{b})}\,D_{b}\Theta[\delta](0)\,,\label{p2}\\
q_{1}(a,b)&=D_{a}D_{b}\ln\Theta[\delta](\textstyle\int^{b}_{a}), \label{q1}\\
q_{2}(a,b)&=\frac{D_{a}\,\Theta[\delta](0)\,D_{b}\,\Theta[\delta](0)}{\Theta[\delta](\int^{b}_{a})^{2}},\label{q2}
\end{align}
where $\delta$ is a non-singular odd characteristics.

\subsection{Real Riemann surfaces}

A Riemann surface Rg is called real if it admits an antiholomorphic involution, denoted by $\tau$.
An anti-holomorphic involution $\tau :\mathcal{R}_{g} \to 
\mathcal{R}_{g}$  satisfies $\tau^2 = id$ and acts on the
local parameter as the complex conjugation. The 
connected components of the set of
fixed points of the anti-involution $\tau$ are called real ovals of $\tau$. We denote 
by $\mathcal{R}_{g}(\R)$ the set of fixed points. 
According to Harnack's inequality \cite{Harnack}, the number $\chi$ of real ovals of a real Riemann surface of
genus $g$ cannot exceed $g+1$: $0\leq \chi\leq g+1$. 
Curves with the maximal number $\chi=g+1$ of real ovals are called M-curves.

The complement $\mathcal{R}_{g}\setminus 
\mathcal{R}_{g}(\R)$ has either one or two connected components. The 
curve $\mathcal{R}_{g}$ is called a \textit{dividing} curve  if 
$\mathcal{R}_{g}\setminus\mathcal{R}_{g}(\R)$ has two components, and 
$\mathcal{R}_{g}$ is called \textit{non-dividing} if 
$\mathcal{R}_{g}\setminus \mathcal{R}_{g}(\R)$ is connected (notice 
that an M-curve is always a dividing curve). In this paper we only 
consider hyperelliptic curves which are discussed  in detail 
section 4.

Let $(\mathbf{\mathcal{A}},\mathbf{\mathcal{B}})$ be a basis of the homology group $H_{1}(\Rs_{g})$. According to Proposition 2.2 in Vinnikov's paper \cite{Vin} (see also \cite{Gross}), there exists a canonical homology basis such that
\begin{equation}
\left(\begin{matrix}
\tau \mathbf{\mathcal{A}}\\
\tau \mathbf{\mathcal{B}}
\end{matrix}\right)
=
\left(\begin{matrix}
\mathbb{I}_{g}&0\\
\mathbb{H}\,\,&-\mathbb{I}_{g}\,\,
\end{matrix}\right)
\left(\begin{matrix}
\mathbf{\mathcal{A}}\\
\mathbf{\mathcal{B}}
\end{matrix}\right), \label{hom basis}
\end{equation}
where $\mathbb{I}_{g}$ is the $g\times g$ unit matrix, and $\mathbb{H}$ is a block diagonal $g\times g$ matrix defined as follows:
\\\\
1) if $\mathcal{R}_{g}(\R)\neq \emptyset$, 
\[\mathbb{H}={\left(\begin{matrix}
0&1&&&&&&\\
1&0&&&&&&\\
&&\ddots&&&&&\\
&&&0&1&&&\\
&&&1&0&&&\\
&&&&&0&&\\
&&&&&&\ddots&\\
&&&&&&&0
\end{matrix}\right)} \quad \text{if $\mathcal{R}_{g}$ is dividing},\]
\[\mathbb{H}={\left(\begin{matrix}
1&&&&&\\
&\ddots&&&&\\
&&1&&&\\
&&&0&&\\
&&&&\ddots&\\
&&&&&0
\end{matrix}\right)} \quad \text{if $\mathcal{R}_{g}$ is non-dividing};\]
rank$(\mathbb{H})=g+1-\chi$ in both cases.
\\\\
2) if $\mathcal{R}_{g}(\R)= \emptyset$, (i.e. the curve does not have 
real ovals), then
\[\mathbb{H}={\left(\begin{matrix}
0&1&&&\\
1&0&&&\\
&&\ddots&&\\
&&&0&1\\
&&&1&0
\end{matrix}\right)}
\quad \text{or} \quad 
\mathbb{H}={\left(\begin{matrix}
0&1&&&&\\
1&0&&&&\\
&&\ddots&&&\\
&&&0&1&\\
&&&1&0&\\
&&&&&0
\end{matrix}\right)};\]
rank$(\mathbb{H})=g$ if $g$ is even, rank$(\mathbb{H})=g-1$ if $g$ is odd.
\\

In what follows we choose a canonical homology basis in $H_{1}(\Rs_{g})$ satisfying (\ref{hom basis}) and 
take $a,b\in\Rs_{g}$ such that $\tau a=b$. Denote by $\ell$ a contour connecting the points $a$
and $b$ which does not intersect the canonical homology basis.
Then the action of $\tau$ on the generators 
$(\mathbf{\mathcal{A}},\mathbf{\mathcal{B}},\ell)$ of the relative homology group $H_{1}(\Rs_{g},\{a,b\})$ is given by  (see \cite{Kalla} for more details)
\begin{equation}
\left(\begin{matrix}
\tau \mathbf{\mathcal{A}}\\
\tau \mathbf{\mathcal{B}}\\
\tau \ell
\end{matrix}\right)
=
\left(\begin{array}{ccc}
\mathbb{I}_{g}&0&0\\
\mathbb{H}\,\,&-\mathbb{I}_{g}\,\,&0\\
\mathbf{N}^{t}&0&-1\,\,\,\,
\end{array}\hspace{-0.3cm}\right)
\left(\begin{matrix}
\mathbf{\mathcal{A}}\\
\mathbf{\mathcal{B}} \\
\ell
\end{matrix}\right), \label{hom basis 1}
\end{equation}
for some $\mathbf{N}\in\Z^{g}$. In the case where $\tau a=a$ and 
$\tau b=b$,  the action of $\tau$ on the generators 
$(\mathbf{\mathcal{A}},\mathbf{\mathcal{B}},\ell)$ of the relative 
homology group $H_{1}(\Rs_{g},\{a,b\})$ reads  (see \cite{Kalla} for more details)
\begin{equation}
\left(\begin{matrix}
\tau \mathbf{\mathcal{A}}\\
\tau \mathbf{\mathcal{B}}\\
\tau \ell
\end{matrix}\right)
=
\left(\begin{array}{ccc}
\mathbb{I}_{g}&0&0\\
\mathbb{H}\,\,&-\mathbb{I}_{g}\,\,&0\\
\mathbf{N}^{t}&\,\,\mathbf{M}^{t}&1
\end{array}\right)
\left(\begin{matrix}
\mathbf{\mathcal{A}}\\
\mathbf{\mathcal{B}} \\
\ell
\end{matrix}\right), \label{hom basis 3}
\end{equation}
where the vectors $\mathbf{N},\,\mathbf{M}\in\Z^{g}$ are related by
\begin{equation}
2\,\mathbf{N}+\mathbb{H}\mathbf{M}=0. \label{NM stable}
\end{equation}

Now let us study the action of $\tau$ on Abelian differentials and the action of the complex conjugation on the theta function with zero characteristics. 
Denote by $\tau^{*}$ the action of $\tau$ lifted to the space of differentials: $\tau^* \omega(p) = \omega(\tau p)$ for any $p\in \mathcal{R}_{g}$. By (\ref{hom basis}) the $\mathcal{A}$-cycles of the homology basis are invariant under $\tau$.  Due to the normalization conditions (\ref{norm hol diff}), this leads to the following action of $\tau$ on the normalized holomorphic differentials: 
\begin{equation}
\overline{\tau^{*}\omega_{j}}=-\,\omega_{j}. \label{diff hol}
\end{equation} 

Let $a,b\in\Rs_{g}$ and denote by $\left\{\mathbf{\mathcal{A}},\mathbf{\mathcal{B}},\mathcal{S}_{b}\right\}$ the generators of the homology group $H_{1}(\Rs_{g}\setminus\{a,b\})$ of the punctured Riemann surface $\Rs_{g}\setminus\{a,b\}$, where $\mathcal{S}_{b}$ is a positively oriented small contour around $b$ such that $\mathcal{S}_{b}\circ\ell=1$.  
It was proved in \cite{Kalla} that in the case where $\tau a = b$, 
the $\mathcal{A}$-cycles in the homology group 
$H_{1}(\Rs_{g}\setminus\{a,b\})$ are stable under $\tau$.
Therefore, by the uniqueness of the normalized differential of the  third kind $\Omega_{b-a}$ which has residue $1$ at $b$ and residue $-1$ at $a$, we get 
\begin{equation}
\overline{\tau^{*}\Omega_{b-a}}=-\,\Omega_{ b - a}. \label{diff 3kind}
\end{equation}
In the case where $\tau a = a$ and $\tau b = b$, Proposition A.2 in \cite{Kalla} shows that the action of $\tau$ on the $\mathcal{A}$-cycles in the homology group $H_{1}(\Rs_{g}\setminus\{a,b\})$ is given by 
\begin{equation}
\tau \mathcal{A}= \mathcal{A}-\mathbf{M} \mathcal{S}_{b},  \label{tau A}
\end{equation}
where $\mathbf{M}$ is defined in (\ref{hom basis 3}). Therefore, by the uniqueness of the differential $\Omega_{b-a}$, we deduce that
\begin{equation}
\overline{\tau^{*}\Omega_{b-a}}=\Omega_{b-a}+\mathbf{M}^{t}\omega, \label{diff 3kind real}
\end{equation}
where $\omega$ denotes the vector of normalized holomorphic differentials. 
From (\ref{hom basis}) and (\ref{diff hol}) we obtain the following reality property for the Riemann matrix $\mathbb{B}$:
\begin{equation}
\overline{\mathbb{B}}= \mathbb{B}-2\mathrm{i}\pi\, \mathbb{H}. \label{matrix B}
\end{equation}
Moreover, according to Proposition 2.3 in \cite{Vin}, for any $\mathbf{z}\in\C^{g}$, relation (\ref{matrix B}) implies
\begin{equation} 
\overline{\Theta(\mathbf{z})}=\kappa\,\Theta(\overline{\mathbf{z}}-\mathrm{i}\pi \,\text{diag}(\mathbb{H})), \label{conj theta}
\end{equation}
where $\text{diag}(\mathbb{H})$ denotes the vector of diagonal 
elements of the matrix $\mathbb{H}$, and $\kappa$ is a root of unity 
which depends on the matrix $\mathbb{H}$ (knowledge of the exact value of $\kappa$ is not needed for our purpose).

\subsection{Action of $\tau$ on the Jacobian and the theta divisor of real Riemann surfaces}

In this subsection, we review known results about the theta divisor of real Riemann surfaces (see \cite{Vin,DN}). Let us choose a canonical homology basis satisfying (\ref{hom basis}) and consider the Jacobian $J:=J(\Rs_{g})$ of a real Riemann surface $\Rs_{g}$.

The anti-holomorphic involution $\tau$ on $\Rs_{g}$ gives rise 
to an anti-holomorphic involution on the Jacobian: 
if $\mathcal{D}:=\mathcal{D}_{2}-\mathcal{D}_{1}$ with $\mathcal{D}_{1}$ and $\mathcal{D}_{2}$ positive divisors on $\Rs_{g}$ (recall that a positive divisor is defined by a finite formal sum of points $\sum_{i} n_{i}\, a_{i}$ with $a_{i}\in\Rs_{g}$ and $n_{i}\in\N$)
then $\tau\,\mathcal{D}$ is the class of the point $(\int_{\tau \mathcal{D}_{1}}^{\tau 
\mathcal{D}_{2}}\omega)=(\int_{\mathcal{D}_{1}}^{\mathcal{D}_{2}}\tau^{*}\omega)$ in the Jacobian. Therefore,  by (\ref{diff hol}) $\tau$ lifts to the anti-holomorphic involution on $J$, denoted also by $\tau$, given by 
\begin{equation}
\tau \zeta=-\overline{\zeta},   \label{3.7}
\end{equation}
for any $\zeta\in J$.

Now consider the following two subsets of the Jacobian
\begin{align}
S_{1}=\{\zeta\in J;\, \zeta+\tau\,\zeta=\mathrm{i}\pi \,\text{diag}(\mathbb{H})\},  \label{S1}\\
S_{2}=\{\zeta\in J;\, \zeta-\tau\,\zeta=\mathrm{i}\pi \,\text{diag}(\mathbb{H})\}, \label{S2}
\end{align}
where the matrix $\mathbb{H}$ was introduced in (\ref{hom basis}).
Below we study their intersections $S_{1}\cap (\Theta)$ and 
$S_{2}\cap (\Theta)$ with the theta divisor $(\Theta)$, the set of zeros of the theta function. 
Let us introduce the notation: the vectors $e_{i}$, $i=1,\ldots,g$ 
with components $e_{ik}=\delta_{ik}$, $\mathbb{B}_{i}=\mathbb{B}\,e_{i}$. 

It is a straightforward computation to prove that the set $S_{1}$ is the disjoint 
union of the tori $T_{v}$ defined by 
\begin{multline}
T_{v}=\{\zeta\in J;\, \zeta=\mathrm{i}\pi\,(\text{diag}(\mathbb{H})/2+v_{1}\,e_{r+1}+\ldots+v_{g-r}\,e_{g}) 
+\beta_{1}\,\text{Re}(\B_{1})+\ldots+\beta_{g}\,\text{Re}(\B_{g})\,, 
\\
\beta_{1},\ldots,\beta_{r}\in\R/2\Z\,,\,\beta_{r+1},\ldots,\beta_{g}\in\R/\Z\}, \label{3.10}
\end{multline}
where $v=(v_{1},\ldots,v_{g-r})\in(\Z/2\Z)^{g-r}$ and $r$ is the rank of the matrix $\mathbb{H}$.
Therefore, the description of the set $S_{1}\cap (\Theta)$ reduces to the study of the sets $T_{v}\cap (\Theta)$.
In the case where $\mathcal{R}_{g}(\R)\neq \emptyset$ and $\mathcal{R}_{g}$ is 
non-dividing, one can see that for all $v$ the torus $T_{v}$ contains a half-period 
corresponding to an odd half-integer characteristics, which yields 
$T_{v}\cap (\Theta)\neq\emptyset$. The same holds for all $v\neq 0$ 
in the case where the curve is dividing or does not have real ovals. 
The following result proved in \cite{Vin} provides a complete description of the sets $T_{v}\cap (\Theta)$ in the case where the curve admits real ovals:

\begin{proposition} 
If $\mathcal{R}_{g}(\R)\neq\emptyset$, then $T_{v}\cap (\Theta)=\emptyset$ if and only if the curve is dividing and $v=0$.  \label{prop div S1}
\end{proposition}

\noindent
In other words, among all curves which admit real ovals, the only torus $T_v$ which does not intersect the theta divisor is the  torus $T_{0}$ corresponding to dividing curves. 
This torus is given by 
\begin{equation}
T_{0}=\{\zeta\in J; \,\zeta=\beta_{1}\,\text{Re}(\B_{1})+\ldots+\beta_{g}\,\text{Re}(\B_{g}),
\,\beta_{1},\ldots,\beta_{r}\in\R/2\Z\,,\,\beta_{r+1},\ldots,\beta_{g}\in\R/\Z \}.  \label{3.11}
\end{equation}

Analogously, it can be checked that the set $S_{2}$ is the disjoint union of the tori 
$\tilde{T}_{v}$ defined by
\begin{multline}
\tilde{T}_{v}=\{\zeta\in J\,;\, \zeta=2\mathrm{i}\pi\left(\alpha_{1}\,e_{1}+\ldots+\alpha_{g}\,e_{g}\right)+(v_{1}/2)\,\B_{r+1}+\ldots+(v_{g-r}/2)\,\B_{g},\\ \label{3.12}
\alpha_{1},\ldots,\alpha_{g}\in\R/\Z\}, 
\end{multline}
where $v=(v_{1},\ldots,v_{g-r})\in(\Z/2\Z)^{g-r}$ and $r$ is the rank of the matrix $\mathbb{H}$. Description of the sets $\tilde{T}_{v}\cap (\Theta)$ in the case where the curve admits real ovals was given in \cite{DN}:

\begin{proposition}
If $\mathcal{R}_{g}(\R)\neq\emptyset$, then $\tilde{T}_{v}\cap (\Theta)=\emptyset$ if and only if the curve is an M-curve and $v=0$.  \label{prop div S2}
\end{proposition}

\section{Algebro-geometric solutions of the Camassa-Holm equation}

In this section we will use Fay's identities to construct solutions to 
the CH equation on hyperelliptic surfaces. For the resulting formulae 
we establish conditions under which we obtain real and smooth solutions. 
In what follows  $\Rs_{g}$ denotes a hyperelliptic curve of genus $g>0$, written as
\begin{equation}
\mu^{2}=\prod_{i=1}^{2g+2}(\lambda-\lambda_{i}),  \label{HypCH}
\end{equation}
where the branch points $\lambda_{i}\in\C$ satisfy the relations $\lambda_{i}\neq \lambda_{j}$ for $i\neq j$. We denote by $\sigma$  the hyperelliptic involution defined by $\sigma(\lambda,\mu)=(\lambda,-\mu)$.
Note that the CH equation can be expressed in the following simple form,
\begin{equation}
m_{t}+u\, m_{x}+2\,m\, u_{x}=0, \label{CH}
\end{equation}
where we put $m:=u-u_{xx}+k$.

\subsection{Identities between theta functions}

In our approach to construct algebro-geometric solutions of the CH 
equation we use the corollaries (\ref{cor Fay1}) and (\ref{cor Fay2}) of Fay's identity.

\begin{proposition} 
Let $a,b\in\Rs_{g}$ such that $\sigma(a)=b$ and let $e\in\Rs_{g}$ be a ramification point, namely, $e=(\lambda_{j},0)$ for some $j\in\{1,\ldots,2g+2\}$. Denote by $g_{1}$ and $g_{2}$ the following functions of the variable $\z\in\C^{g}$:
\begin{equation}
g_{1}(\z)=\frac{\Theta\hspace{-2pt}\left(\z+\frac{\mr}{2}\right)}{\Theta(\z)}, \qquad  g_{2}(\z)=\frac{\Theta\hspace{-2pt}\left(\z-\frac{\mr}{2}\right)}{\Theta(\z)}, \label{giCH}
\end{equation}
where $\mr=\int_{a}^{b}\omega$ and $\omega$ is the vector of normalized holomorphic differentials. Then the two following identities hold:
\begin{equation}
D_{b}D_{e}\ln\frac{g_{1}}{g_{2}}=-\,\frac{p_{2}}{g_{1}g_{2}}\,D_{b}\ln g_{1}g_{2}, \label{cor1CH}
\end{equation}
\begin{equation}
D_{b}D_{e}\ln(g_{1}g_{2})=\frac{\tilde{q}_{2}}{\tilde{p}_{2}}\,\frac{1}{g_{1}g_{2}}\left(D_{b}\ln\frac{g_{1}}{g_{2}}-2\,\tilde{p}_{1}\right)-2\,\tilde{q}_{2}\,g_{1}g_{2}.\label{cor2CH}
\end{equation}
Here we used the notation:
\begin{equation}
p_{2}=p_{2}(b,e,a),\qquad \tilde{p}_{i}=p_{i}(e,b,a), \qquad \tilde{q}_{2}=q_{2}(b,e), \label{notCH}
\end{equation}
where the scalars $q_{2}(.,.)$ and $p_{i}(.,.,.),\,i=1,2,$  are 
defined in (\ref{q2}) and (\ref{p1}), (\ref{p2}); $D_{b}$ (respectively $D_{e}$) denotes the directional derivative along the vector $\mathbf{V}_{b}$ (respectively $\mathbf{V}_{e}$) defined in (\ref{exp hol diff}). 
\end{proposition}

\begin{proof}
Under the changes of variables $(a,b,c)\rightarrow(b,e,a)$ and $\z\rightarrow \z-\mr/2$, identity (\ref{cor Fay1}) becomes
\begin{equation}
D_{e}\ln \frac{g_{1}}{g_{2}}=p_{1}+\frac{p_{2}}{g_{1}g_{2}}, \label{proof1CH}
\end{equation}
where we used the notation $p_{i}=p_{i}(b,e,a)$ for $i=1,2$.
Here we used the fact that 
$\int_{a}^{e}\omega=\int_{e}^{b}\omega=\mr/2$, according to the 
action of $\sigma^{*}$ (the action of $\sigma$ lifted to the space of 
one-forms) on the normalized holomorphic differentials $\omega_{j}$:
\begin{equation}
\sigma^{*}\omega_{j}=-\,\omega_{j}, \qquad j=1,\ldots,g. \label{sig}
\end{equation}
Applying the differential operator $D_{b}$ to equation (\ref{proof1CH}) one gets:
\[D_{b}D_{e}\ln\frac{g_{1}}{g_{2}}=-\,p_{2}\,\frac{D_{b}(g_{1}g_{2})}{(g_{1}g_{2})^{2}}=-\,\frac{p_{2}}{g_{1}g_{2}}\,D_{b}\ln g_{1}g_{2},\]
which proves (\ref{cor1CH}). To prove (\ref{cor2CH}), consider the change of variables $(a,b,c)\rightarrow(e,b,a)$ in (\ref{cor Fay1}), which leads to
\begin{equation}
D_{b}\ln g_{1}=\tilde{p}_{1}+\tilde{p}_{2}\,g_{2}\,\frac{\Theta(\z+\mr)}{\Theta\hspace{-2pt}\left(\z+\frac{\mr}{2}\right)}. \nonumber
\end{equation}
Changing $\z$ to $-\z$ in the last equality, one gets
\begin{equation}
D_{b}\ln g_{2}=-\,\tilde{p}_{1}-\tilde{p}_{2}\,g_{1}\,\frac{\Theta(\z-\mr)}{\Theta\hspace{-2pt}\left(\z-\frac{\mr}{2}\right)}. \nonumber
\end{equation}
From these two identities, it can be deduced that
\begin{align}
\frac{\Theta(\z+\mr)}{\Theta\hspace{-2pt}\left(\z+\frac{\mr}{2}\right)}&=(\tilde{p}_{2}\,g_{2})^{-1}\,(D_{b}\ln g_{1}-\tilde{p}_{1}),\label{proof4CH}
\\
 \frac{\Theta(\z-\mr)}{\Theta\hspace{-2pt}\left(\z-\frac{\mr}{2}\right)}&=-\,(\tilde{p}_{2}\,g_{1})^{-1}\,(D_{b}\ln g_{2}+\tilde{p}_{1})\label{proof5CH}.
\end{align}
Moreover, since
\[D_{b}D_{e}\ln(g_{1}g_{2})=D_{b}D_{e}\ln\Theta\hspace{-2pt}\left(\z+\frac{\mr}{2}\right)+D_{b}D_{e}\ln\Theta\hspace{-2pt}\left(\z-\frac{\mr}{2}\right)-2\,D_{b}D_{e}\ln\Theta(\z),\]
using (\ref{cor Fay2}) one gets
\begin{equation}
D_{b}D_{e}\ln(g_{1}g_{2})=\frac{\tilde{q}_{2}}{g_{1}}\,\frac{\Theta(\z+\mr)}{\Theta\hspace{-2pt}\left(\z+\frac{\mr}{2}\right)}+\frac{\tilde{q}_{2}}{g_{2}}\,\frac{\Theta(\z-\mr)}{\Theta\hspace{-2pt}\left(\z-\frac{\mr}{2}\right)}-2\,\tilde{q}_{2}\,g_{1}g_{2},  \label{proof6CH}
\end{equation}
which by (\ref{proof4CH}) and (\ref{proof5CH}) leads to (\ref{cor2CH}).
\end{proof}

\noindent
With identities (\ref{cor1CH}) and (\ref{cor2CH}) we are now able to construct theta-functional solutions of the CH equation:

\begin{theorem} Let $a,b\in\Rs_{g}$ such that 
$\sigma(a)=b$, and let $e\in\Rs_{g}$ be a ramification point. Denote by 
$\ell$ an oriented contour between $a$ and $b$ which contains the point $e$. Assume that $\ell$ does not 
cross cycles of the canonical homology basis. Choose arbitrary 
constants $\mathbf{d}\in\C^{g}$ and $k,\,\zeta\in\C$, and put
\begin{equation}
\alpha_{1}=p_{1}(b,e,a), \qquad \alpha_{2}=2\,p_{1}(e,b,a)+k, \label{alpha CH}
\end{equation}
where the function $p_{1}$ is defined in (\ref{p1}). Let $y(x,t)$ be an implicit function of the variables $x,\,t\in\R$ defined by
\begin{equation}
x+\alpha_{1}\,y+\alpha_{2}\,t+\zeta=\ln \frac{\Theta\hspace{-2pt}\left(\mathbf{Z}-\mathbf{d}+\frac{\mr}{2}\right)}{\Theta\hspace{-2pt}\left(\mathbf{Z}-\mathbf{d}-\frac{\mr}{2}\right)},  \label{implicit equ}
\end{equation}
where $\mr=\int_{\ell}\omega$. Here the vector $\mathbf{Z}$ is given by
\begin{equation}
\mathbf{Z}(x,t)=\mathbf{V}_{e}\,y(x,t)+\mathbf{V}_{b}\,t, \label{Z CH}
\end{equation}
where the vectors $\mathbf{V}_{e}$ and $\mathbf{V}_{b}$ are defined in (\ref{exp hol diff}). Then the following function of the variables $x$ and $t$ is solution of the CH equation:
\begin{equation}
u(x,t)=D_{b}\ln \frac{\Theta\hspace{-2pt}\left(\mathbf{Z}-\mathbf{d}+\frac{\mr}{2}\right)}{\Theta\hspace{-2pt}\left(\mathbf{Z}-\mathbf{d}-\frac{\mr}{2}\right)}-\alpha_{2}. \label{sol u CH} \\
\end{equation}
Here $D_{b}$ denotes the directional derivative along the vector $\mathbf{V}_{b}$.  \label{prop sol CH}
\end{theorem}
Note that the function $y(x,t)$ is the same function as introduced in 
    \cite{AFY2}.

\begin{proof}
Let $\beta,\delta\in\C$ and $\alpha_{1},\alpha_{2}\in\C$ be arbitrary constants. Let us look for solutions $u$ of CH having the form
\begin{equation}
u(x,t)=\beta\,D_{b}\ln \frac{\Theta\hspace{-2pt}\left(\mathbf{Z}-\mathbf{d}+\frac{\mr}{2}\right)}{\Theta\hspace{-2pt}\left(\mathbf{Z}-\mathbf{d}-\frac{\mr}{2}\right)}+\delta=\beta\,D_{b}\ln \frac{g_{1}}{g_{2}}+\delta, \label{u proof CH}
\end{equation}
where $\mathbf{Z}(x,t)$ is defined in (\ref{Z CH}), and the functions $g_{1}, g_{2}$ were 
introduced in (\ref{giCH}) with $\z=\mathbf{Z}(x,t)-\mathbf{d}$.
By (\ref{implicit equ}), the derivative with respect to the variable $x$ of the implicit function $y(x,t)$ is given by
\begin{equation}
y_{x}=\left(D_{e}\ln\frac{g_{1}}{g_{2}}-\alpha_{1}\right)^{-1}. \nonumber
\end{equation}
With $\alpha_{1}=p_{1}(b,e,a)$, relation (\ref{proof1CH}) implies
\begin{equation}
y_{x}=\frac{g_{1}g_{2}}{p_{2}}.\label{y_x bis}
\end{equation}
Analogously it can be checked that 
\begin{equation}
y_{t}=-\,y_{x}\left(\frac{u}{\beta}-\frac{\delta}{\beta}-\alpha_{2}\right). \label{y_t}
\end{equation}
Now let us express the function $m(x,t)=u-u_{xx}+k$ introduced in 
(\ref{CH}) in terms of the functions $g_{1}$ and $g_{2}$ of 
(\ref{giCH}). By (\ref{cor1CH}) and (\ref{y_x bis}), the first derivative of the function $u$ (\ref{u proof CH}) with respect to the variable $x$ is given by
\begin{equation}
u_{x}=-\,\beta\,D_{b}\ln (g_{1}g_{2}).  \label{u_x}
\end{equation}
By  (\ref{cor2CH})  and (\ref{y_x bis}) we obtain for the second derivative of $u$ with respect to $x$:
\begin{equation}
u_{xx}=\beta\left(D_{b}\ln\frac{g_{1}}{g_{2}}-2\,\tilde{p}_{1}\right)-2\beta\,\tilde{p}_{2}\,(g_{1}g_{2})^{2}\,;  \label{u_xx}
\end{equation}
here we used the identity $\tilde{q}_{2}=-\,\tilde{p}_{2}\,p_{2}$  relating the scalars $\tilde{q}_{2},\tilde{p}_{2}$ and $p_{2}$ defined in (\ref{notCH}).
Therefore, with (\ref{u proof CH}) and (\ref{u_xx}), the function $m$ reads 
\begin{equation}
m(x,t)=\delta+k+2\beta\,\tilde{p}_{1}+2\beta\,\tilde{p}_{2}\,(g_{1}g_{2})^{2}. \label{m}
\end{equation}
Taking the derivative of $m$ with respect to $x$, and the derivative of $m$ with respect to $t$, one gets respectively:
\begin{align}
m_{x}(x,t)&=4\beta\,\tilde{p}_{2}\,(g_{1}g_{2})^{2}\,y_{x}\,D_{e}\ln(g_{1}g_{2}), \label{m_x}\\
m_{t}(x,t)&=4\beta\,\tilde{p}_{2}\,(g_{1}g_{2})^{2}\,\left(y_{t}\,D_{e}\ln(g_{1}g_{2})+D_{b}\ln(g_{1}g_{2})\right). \label{m_t}
\end{align}
Therefore, substituting the functions (\ref{u proof CH}), 
(\ref{u_x}), (\ref{m_x}) and (\ref{m_t}) in the left-hand side of the CH equation 
(\ref{CH}) we obtain
\begin{eqnarray}
2\,\tilde{p}_{2}\,D_{e}\ln(g_{1}g_{2})\,y_{x}\,(g_{1}g_{2})^{2}\,\left[u\left(1-\frac{1}{\beta}\right)+\frac{\delta}{\beta}+\alpha_{2}\right]& \nonumber\\
-\,D_{b}\ln(g_{1}g_{2})\,\left(\delta+k+2\beta\,\tilde{p}_{1}+2\,\tilde{p}_{2}\,(g_{1}g_{2})^{2}\,(\beta-1)\right)\,=\,0. \nonumber
\end{eqnarray}
This equality holds for $\beta=1,\,\delta=-\,2\,\tilde{p}_{1}-k$ and $\alpha_{2}=-\,\delta,$
which completes the proof.
\end{proof}

\subsection{Real-valued solutions and smoothness conditions}

In this subsection, we  identify  real-valued 
and smooth solutions of the CH equation among the solutions given in 
Proposition 3.2. 
Let us first recall that hyperelliptic M-curves of genus $g$ can be given by the equation
\begin{equation}
\mu^{2}= \prod_{i=1}^{2g+2} (\lambda-\lambda_{i}), \label{hyp real1}
\end{equation}
where the branch points $\lambda_{i}$ are real and satisfy $\lambda_{i}\neq \lambda_{j}$ if $i\neq j$. 
On such a curve, we can define 
two anti-holomorphic involutions $\tau_{1}$ and $\tau_{2}$, given 
respectively by 
$\tau_{1}(\lambda,\mu)=(\overline{\lambda},\overline{\mu})$ and 
$\tau_{2}(\lambda,\mu)=(\overline{\lambda},-\overline{\mu})$. Let us show that the curve (\ref{hyp real1}) is an M-curve with respect to both anti-involutions $\tau_{1}$ and $\tau_{2}$.
In the case where $\lambda_{i}\in\R$ satisfy $\lambda_{1}<\ldots<\lambda_{2g+2}$, it can be seen that projections of real ovals of $\tau_{1}$  on the $\lambda$-plane coincide with  the intervals $[\lambda_{2g+2},\lambda_{1}],[\lambda_{2},\lambda_{3}],\ldots,[\lambda_{2g},\lambda_{2g+1}]$, whereas projections of real ovals of $\tau_{2}$ on the $\lambda$-plane coincide with  the intervals $[\lambda_{1},\lambda_{2}],\ldots,[\lambda_{2g+1},\lambda_{2g+2}]$. Hence the curve (\ref{hyp real1}) has the maximal number $g+1$ of real ovals with respect to both anti-involutions $\tau_{1}$ and $\tau_{2}$.

Now assume that $\Rs_{g}$ is a real 
hyperelliptic curve which admits real ovals with respect to an anti-holomorphic involution $\tau$.
Let us choose a homology 
basis satisfying (\ref{hom basis}). Recall that $\Rs_{g}(\R)$ denotes the set of fixed points of the anti-holomorphic involution $\tau$.

The following propositions provide
reality and smoothness conditions for the solutions $u(x,t)$ (\ref{sol u 
CH}) in the case where the points $a$ and $b$ are stable under 
$\tau$, and in the case where $\tau a=b$. It is proved 
that, for fixed $t_{0}\in\R$, the function $u(x,t_{0})$ is smooth 
with respect to the real variable $x$ when $a$ and $b$ are stable 
under $\tau$. In the case where $\tau a=b$,  the function 
$u(x,t_{0})$ is either smooth, or it has cusp-like singularities. \index{cusp-like singularity}

\begin{proposition} Assume that $\Rs_{g}$ is a hyperelliptic M-curve of genus $g$,  and denote by $e\in\Rs_{g}(\R)$ one of its ramification points. Let $a,b\in\Rs_{g}(\R)$ such that $\sigma(a)=b$.
For any $c\in\{a,b,e\}$, choose a local parameter $k_{c}$ such that $\overline{k_{c}(\tau 
p)}=k_{c}(p)$ for any point $p$ in a neighborhood of $c$.  Denote by $\ell$ an oriented contour 
between $a$ and $b$ containing point $e$ which does not intersect cycles of the canonical homology basis. Choose $\ell$ such that the closed path $\tau \ell-\ell$ is homologous to zero in $H_{1}(\Rs_{g})$. Take $\mathbf{d}\in \mathrm{i}\R^{g}$ and $k\in\R$. Choose $\zeta\in\C$ in (\ref{implicit equ}) such that  $\text{Im}(\zeta)=\arg\left\{\ln\left(\frac{\Theta(\mathbf{d}+\mr/2)}{\Theta(\mathbf{d}-\mr/2)}\right)\right\}$. Then solutions $u(x,t)$ of the CH equation given in (\ref{sol u CH}) are real-valued, and for fixed $t_{0}\in\R$, the function $u(x,t_{0})$ is smooth with respect to the real variable $x$. \label{prop real sol CH1}
\end{proposition}

\begin{proof}
Let us check that under the conditions of the proposition, the function $u(x,t)$ (\ref{sol u CH}) is real-valued. Let us fix $y,t\in\R$. First of all, invariance with respect to the anti-involution $\tau$ of the points $e$ and $b$ implies 
\begin{equation}
\overline{\mathbf{Z}}=-\,\mathbf{Z}, \label{real Z CH}
\end{equation}
where the vector $\mathbf{Z}$ is defined in (\ref{Z CH}).
In fact, using the expansion (\ref{exp hol diff}) of the normalized holomorphic differentials $\omega_{j}$ near $c\in\{e,b\}$, one gets
\[\overline{\tau^{*} \omega_{j}}(c)(p)=\left(\,\overline{V_{c,j}}+\overline{W_{c,j}}\,k_{c}(p)+o\left(k_{c}(p)^{2}\right)\right)\ud k_{c}(p),\]
for any point $p$ in a neighborhood of $c$.
Then by (\ref{diff hol}), the vectors $\mathbf{V}_{e}$ and $\mathbf{V}_{b}$ appearing in the vector $\mathbf{Z}$ are purely imaginary,
which leads to (\ref{real Z CH}). 
Moreover, since the closed contour $\tau \ell-\ell$ is homologous to zero in $H_{1}(\Rs_{g})$, from (\ref{diff hol}) one gets
\begin{equation}
\overline{\mathbf{r}}=-\,\mathbf{r}. \label{r CH}
\end{equation}
For arbitrary points $a_{1},a_{2},a_{3}\in\Rs_{g}$, using the representation of the differential $\Omega_{a_{3}-a_{1}}$  in terms 
of multi-dimensional theta functions (see, for instance, \cite{BBEIM}),  one gets:
\begin{equation}
\frac{\Omega_{a_{3}-a_{1}}(p)}{\ud k_{a_{2}}(p)}\Big|_{p=a_{2}}=p_{1}(a_{1},a_{2},a_{3}). \label{p1 Fay}
\end{equation} 
We deduce that the scalars $p_{1}(b,e,a)$ and $p_{1}(e,b,a)$ 
appearing respectively in $\alpha_{1}$ and $\alpha_{2}$ (see 
(\ref{alpha CH})) satisfy
\begin{equation}
\frac{\Omega_{a-b}(p)}{\ud k_{e}(p)}\Big|_{p=e}=p_{1}(b,e,a), \qquad   \frac{\Omega_{a-e}(p)}{\ud k_{b}(p)}\Big|_{p=b}=p_{1}(e,b,a).\label{pi diff}
\end{equation}
Therefore, since the points $a,b,e$ are stable under $\tau$,  from (\ref{diff 3kind real}) it can be deduced that $p_{1}(b,e,a)$ and $p_{1}(e,b,a)$ are real, which involves
\begin{equation}
\overline{\alpha_{1}}=\alpha_{1}, \qquad \overline{\alpha_{2}}=\alpha_{2}.  \label{conj alpha}
\end{equation}
Let $y,t\in\R$ and denote by $h$ the function
\begin{equation}
h(y,t)= \frac{\Theta\hspace{-2pt}\left(\mathbf{Z}-\mathbf{d}+\frac{\mr}{2}\right)}{\Theta\hspace{-2pt}\left(\mathbf{Z}-\mathbf{d}-\frac{\mr}{2}\right)}. \label{fct h}
\end{equation} 
By (\ref{implicit equ}), $x$ is a real-valued function of the real variables $y$ and $t$ if the function $h$ is real and has a constant sign, and if we choose $\text{Im}(\zeta)=\arg\{\ln(h(0,0))\}$. From (\ref{conj theta}), (\ref{real Z CH}) and (\ref{r CH}) we deduce that
\begin{equation}
\overline{ h(y,t)}
=\frac{\Theta\hspace{-2pt}\left(\mathbf{Z}+\overline{\mathbf{d}}+\frac{\mr}{2}+\mathrm{i}\pi\,\text{diag}(\mathbb{H})\right)}{\Theta\hspace{-2pt}\left(\mathbf{Z}+\overline{\mathbf{d}}-\frac{\mr}{2}+\mathrm{i}\pi\,\text{diag}(\mathbb{H})\right)}. \label{real1 CH}
\end{equation}
Let us choose a vector $\mathbf{d}\in \C^{g}$ such that
\begin{equation}
 \overline{\mathbf{d}}= -\,\mathbf{d}- \mathrm{i}\pi \,\text{diag}(\mathbb{H}) + 2\mathrm{i}\pi\mathbf{T}+\B\mathbf{L} \label{real d1 CH}  \nonumber
\end{equation}
for some vectors $\mathbf{T},\mathbf{L}\in\Z^{g}$. Reality of the 
vector  $\overline{\mathbf{d}}+\mathbf{d}$ together with (\ref{matrix 
B}) implies 
\begin{equation}
\mathbf{d}=\frac{1}{2}\,\mathrm{Re}(\mathbb{B})\,\mathbf{L}+\mathrm{i}\,\mathbf{d}_{I} \label{real d2 CH}
\end{equation}
for some $\mathbf{d}_{I}\in\R^{g}$ and the relation $2\,\mathbf{T}+ 
\mathbb{H}\mathbf{L}=\text{diag}(\mathbb{H})$ for  $\mathbf{T}$ and $\mathbf{L}$.
For this choice of the vector $\mathbf{d}$, (\ref{real1 CH}) becomes
\begin{equation}
\overline{ h(y,t)}
=\frac{\Theta\hspace{-2pt}\left(\mathbf{Z}-\mathbf{d}+\frac{\mr}{2}\right)}{\Theta\hspace{-2pt}\left(\mathbf{Z}-\mathbf{d}-\frac{\mr}{2}\right)}\,\exp\{-\left\langle \mr,\mathbf{L}\right\rangle\}, \nonumber
\end{equation}
where we used the quasi-periodicity property  (\ref{per theta}) of the theta function. Therefore, the function $h$ is real if $\mathbf{L}=0$, that is $\mathbf{d}\in\mathrm{i}\R^{g}$. Now let us check that $h$ has a constant sign with respect to $y,t\in\R$. Since $\mathbf{Z}-\mathbf{d}\pm\frac{\mr}{2}\in\mathrm{i}\R^{g}$, by Proposition \ref{prop div S2} the functions $\Theta(\mathbf{Z}-\mathbf{d}\pm\frac{\mr}{2})$ of the real variables $y$ and $t$ do not vanish if the hyperelliptic curve is an M-curve, i.e., if all branch points in (\ref{HypCH}) are real.
Hence $h$ is a real continuous non vanishing function with respect to the real variables $y$ and $t$, which means it has a constant sign. Therefore, $x$ is a real-valued function of $y$ and $t$ if the constant $\zeta$ in (\ref{implicit equ}) is chosen such that   $\text{Im}(\zeta)=\arg\{\ln(h(0,0))\}$.
It is straightforward to see that the solution $u$ (\ref{sol u CH}) is a real-valued function of the real variables $y$ and $t$, and then is real-valued with respect to the real variables $x$ and $t$.

Now fix $t_{0}\in\R$ and let us study smoothness conditions for the 
function $u(x,t_{0})$ with respect to the variable $x$. First let us 
check that the solution $u$ 
(\ref{sol u CH}) is a smooth function of the real variable $y$ 
if it does not have singularities. Since the theta 
function is entire, singularities of the solution $u$ are located at 
the zeros of its denominator. As we saw in the previous paragraph,  if the curve is an M-curve and $\mathbf{d}\in \mathrm{i}\R^{g}$, the functions $\Theta(\mathbf{Z}-\mathbf{d}\pm\frac{\mr}{2})$ and $\Theta(\mathbf{Z}-\mathbf{d})$ do not vanish. In this case, the function $u$ is smooth with respect to the real variable $y$. Now let us prove that $u$ is smooth with respect to the real variable $x$. By (\ref{implicit equ}), the function $x(y)$ is smooth. Moreover, it can be seen from (\ref{implicit equ}) and (\ref{proof1CH}) that
\begin{equation}
x_{y}(y)=\frac{p_{2}}{g_{1}(\mathbf{Z}-\mathbf{d})\,g_{2}(\mathbf{Z}-\mathbf{d})}=p_{2}\,\frac{\Theta(\mathbf{Z}-\mathbf{d})^{2}}{\Theta\hspace{-2pt}\left(\mathbf{Z}-\mathbf{d}+\frac{\mr}{2}\right)\Theta\hspace{-2pt}\left(\mathbf{Z}-\mathbf{d}-\frac{\mr}{2}\right)}.   \label{x_y real}
\end{equation}
Since the functions $\Theta(\mathbf{Z}-\mathbf{d}\pm\frac{\mr}{2})$ and $\Theta(\mathbf{Z}-\mathbf{d})$ do not vanish, we deduce that $x(y)$ is a strictly monotonic real function, and thus the inverse function $y(x)$ has the same property. Therefore, the function $u(x,t_{0})=u(y(x))$ is a smooth real-valued  function with respect to the real variable $x$.
\end{proof}

\noindent
Now let us study real-valuedness and smoothness of the solutions in the case where $\tau a=b$.

\begin{proposition} Assume that $\Rs_{g}$ is a hyperelliptic M-curve of genus $g$,  and denote by $e\in\Rs_{g}(\R)$ one of its ramification points. Let $a,b\in\Rs_{g}$ such that $\sigma(a)=b$ and assume that $\tau a=b$. 
Choose local parameters such that $\overline{k_{b}(\tau p)}=k_{a}(p)$ for any point $p$ in a neighborhood of $a$, and $\overline{k_{e}(\tau 
p)}=-\,k_{e}(p)$ for any $p$ lying in a neighborhood of $e$.
Denote by $\ell$ an oriented contour between $a$ and $b$ containing point $e$, which does not intersect cycles of the 
canonical homology basis. Assume that $\mathbf{N}=2\mathbf{L}$ for some
$\mathbf{L}\in \Z^{g}$, where $\mathbf{N}\in \Z^{g}$ is defined in (\ref{hom basis 1}). Take $k\in\R$ and define 
$ \mathbf{d}=\mathbf{d}_{R}+\frac{\mathrm{i}\pi}{2}\,\mathbf{N}$ 
 for some $\mathbf{d}_{R}\in\R^{g}$. Choose $\zeta\in\C$ in (\ref{implicit equ}) such that
$\text{Im}(\zeta)=\arg\left\{\ln\left(\frac{\Theta(\mathbf{d}+\mr/2)}{\Theta(\mathbf{d}-\mr/2)}\right)\right\}$. Then solutions $u$ (\ref{sol u CH}) 
of the  CH equation are real-valued. Moreover, for fixed $t_{0}\in\R$, the function $u(x,t_{0})$ is smooth with respect to the real variable $x$ in the case where $\mathbf{N}=0$, otherwise it has   
an infinite number of singularities of the type 
$O\left((x-x_{0})^{\frac{2n}{2n+1}}\right)$ for some 
$n\in\N\setminus\{0\}$ and $x_{0}\in\R$, i.e.,
cusps.  \label{prop real sol CH2}
\end{proposition}

\begin{proof}
Analogously to the case where $a$ and $b$ are stable under $\tau$, 
let us prove that solutions $u$ (\ref{sol u CH}) are real-valued.  Fix $y,t\in\R$.
First let us check that the vector $\mathbf{Z}$ (\ref{Z CH}) satisfies:
\begin{equation}
\overline{\mathbf{Z}}=\mathbf{Z}. \label{vectZ CH conj}
\end{equation}
From (\ref{exp hol diff}) and (\ref{diff hol}) one gets $\overline{\mathbf{V}_{a}}=-\mathbf{V}_{b}$ and $\overline{\mathbf{V}_{e}}=\mathbf{V}_{e}$. 
Moreover, the vectors $\mathbf{V}_{a}$ and $\mathbf{V}_{b}$ satisfy 
$\mathbf{V}_{a}+\mathbf{V}_{b}=0$ because of  (\ref{sig}); thus we 
have  $\overline{\mathbf{V}_{b}}=\mathbf{V}_{b}$ and $\overline{\mathbf{V}_{e}}=\mathbf{V}_{e}$ 
which proves (\ref{vectZ CH conj}).
By (\ref{diff 3kind}) and (\ref{pi diff}) it can be deduced that $\alpha_{1}$ and $\alpha_{2}$ (\ref{alpha CH}) satisfy 
\begin{equation}
\overline{\alpha_{1}}=\alpha_{1}, \qquad \overline{\alpha_{2}}=\alpha_{2}.  \label{alpha conj bis}
\end{equation}
Moreover, from (\ref{diff hol}) and (\ref{hom basis 1}) one gets:
\begin{equation}
\overline{\mathbf{r}}=\mathbf{r}- 2\mathrm{i}\pi \mathbf{N}, \label{r CH bis}
\end{equation}
where $\mathbf{N}\in\Z^{g}$ is defined in (\ref{hom basis 1}).
Let us check that $x$ (\ref{implicit equ}) is a real-valued function 
of the real variables $y$ and $t$. By (\ref{alpha conj bis}), this holds if the function $h$ (\ref{fct h}) is real and has a constant sign with respect to the real variables $y$ and $t$, and if we choose $\text{Im}(\zeta)=\arg\{\ln(h(0,0))\}$.  
By (\ref{conj theta}), (\ref{vectZ CH conj}) and (\ref{r CH bis}) it follows that
\begin{equation}
\overline{h(y,t)}=\frac{\Theta\hspace{-2pt}\left(\mathbf{Z}-\overline{\mathbf{d}}+\frac{\mr}{2}+\mathbf{p}\right)}{\Theta\hspace{-2pt}\left(\mathbf{Z}-\overline{\mathbf{d}}-\frac{\mr}{2}+\mathbf{p}\right)}, \label{real1 CH bis}
\end{equation}
where $\mathbf{p}=-\,\mathrm{i}\pi\mathbf{N}-\mathrm{i}\pi 
\,\text{diag}(\mathbb{H})$. Let us choose the vector $\mathbf{d}\in\C^{g}$ such that
\[\overline{\mathbf{d}}\equiv\mathbf{d}+\mathbf{p} \pmod{2\mathrm{i}\pi \Z^{g}+ \mathbb{B}\Z^{g}}, \]
which is, since $\overline{\mathbf{d}}-\mathbf{d}$ and $\mathbf{p}$ are purely imaginary, equivalent to $\overline{\mathbf{d}}= \mathbf{d}+\mathbf{p} + 2\mathrm{i}\pi\mathbf{T},$ for some $\mathbf{T}\in\Z^{g}$. 
Here we used the action (\ref{matrix B}) of the complex conjugation on the Riemann matrix $\B$, and the fact that $\B$ has a negative definite real part.
 Hence, the vector $\mathbf{d}$ can be written as 
\begin{equation}
 \mathbf{d}=\mathbf{d}_{R}+\frac{\mathrm{i}\pi}{2}(\mathbf{N}+\text{diag}(\mathbb{H})-2\,\mathbf{T}), \label{real d1 CH bis}
\end{equation}
 for some $\mathbf{d}_{R}\in\R^{g}$ and $\mathbf{T}\in\Z^{g}$.
For this choice of the vector  $\mathbf{d}$, by (\ref{real1 CH bis}) the function $h$ is a real-valued function of the real variables $y$ and $t$. Now let us study in which cases the function $h$ has a constant sign. The sign of the function $h$ is constant with respect to $y$ and $t$ if the functions $\Theta(\mathbf{Z}-\mathbf{d}\pm\frac{\mr}{2})$ do not vanish.
By (\ref{vectZ CH conj}), (\ref{r CH bis}) and (\ref{real d1 CH bis}), the vectors $\mathbf{Z}-\mathbf{d}\pm\frac{\mr}{2}$ belong 
to the set $S_{1}$ introduced in (\ref{S1}). Hence by Proposition \ref{prop div S1}, the functions $\Theta(\mathbf{Z}-\mathbf{d}\pm\frac{\mr}{2})$ do not vanish if the hyperelliptic curve is 
dividing (in this case $\text{diag}(\mathbb{H})=0$), and if the 
arguments $\mathbf{Z}-\mathbf{d}\pm\frac{\mr}{2}$ in the theta 
function are real (modulo $2\mathrm{i}\pi\Z^{g}$). The vector  
$\mathbf{Z}-\mathbf{d}+\frac{\mr}{2}$ is real if $\mathbf{T}=0$ in 
(\ref{real d1 CH bis}). With this choice of the vector $\mathbf{T}$, 
the imaginary part of the vector 
$\mathbf{Z}-\mathbf{d}-\frac{\mr}{2}$ equals 
$-\mathrm{i}\pi\mathbf{N}$. Therefore, the vector  
$\mathbf{Z}-\mathbf{d}-\frac{\mr}{2}$ is real modulo 
$2\mathrm{i}\pi\Z^{g}$ if all components of the vector $\mathbf{N}$ 
are even. To summarize, the function $h(y,t)$ defined in (\ref{fct 
h}) is a real-valued function with constant sign if the hyperelliptic 
curve is dividing (i.e., all ramification points are stable under 
$\tau$, since the ramification point $e$ is stable under $\tau$ and 
since dividing curves have either only real branch points, or 
pairwise conjugate ones), if $\mathbf{T}=0$ and $\mathbf{N}=2\mathbf{L}$ for some $\mathbf{L}\in\Z^{g}$, where vector $\mathbf{N}\in\Z^{g}$ is defined in (\ref{hom basis 1}). Analogously to the proof of Proposition \ref{prop real sol CH1}, we conclude that $x$ is a real-valued continuous function of the real variables $y$ and $t$, and thus solutions $u(x,t)$ (\ref{sol u CH}) are real-valued functions of the real variables $x$ and $t$.

Now let us study smoothness conditions for fixed $t_{0}\in\R$. Notice 
that the function $u(y)$ (\ref{sol u CH}) is a smooth function of the real variable  $y$ since the denominator does not vanish, as we have seen before. Put $\z=\mathbf{Z}-\mathbf{d}$. Let us consider the function $x_{y}(y)$ given in (\ref{x_y real}) in both cases: $\mathbf{N}=0$ and $\mathbf{N}\neq 0$.\\
- If $\mathbf{N}=0$, the function $x_{y}(y)$ does not vanish, since in this case $\z\in \R^{g}$ which implies that the function $\Theta(\z)$ does not vanish.
Hence, analogously to the case where $a$ and $b$ 
are stable under $\tau$, for fixed $t_{0}\in\R$, the function $u(x,t_{0})$ is smooth with respect to the real variable $x$. 
\\
- If $\mathbf{N}\neq 0$, the function $\Theta(\z)$ vanishes when $\z$ 
belongs to the theta divisor. Fix $x_{0},t_{0}\in\R$ and denote by $\z_{0}$ and $y_{0}$ 
the corresponding values of $\z$ and $y$. Assume that 
$\z_{0}$ is a zero of the theta function of order $n\geq 1$. Then by 
(\ref{x_y real}), the function $x_{y}(y)$ has a zero at $y_{0}$ of order $2n$. It follows that function $x(y)-x(y_{0})$ has a zero of order $2n+1$ at $y_{0}$, and then
\begin{equation}
y(x)-y_{0}=O\left((x-x_{0})^{\frac{1}{2n+1}}\right).  \label{y(x)}
\end{equation}
On the other hand, it can be seen from (\ref{cor1CH}) that
\begin{equation}
u_{y}(y)=p_{2}\,\frac{\Theta(\z)}{\Theta\hspace{-2pt}\left(\z+\frac{\mr}{2}\right)\Theta\hspace{-2pt}\left(\z-\frac{\mr}{2}\right)}\,\left[\Theta(\z)\,\psi(\z)-2\,D_{b}\Theta(\z)\right], \label{u_{y}(y)}
\end{equation}
where 
$\psi(\z)=D_{b}\ln\left(\Theta(\z+\tfrac{\mr}{2})\,\Theta(\z-\tfrac{\mr}{2})\right)$. Identity (\ref{u_{y}(y)}) implies that function $u_{y}(y)$ has a zero at $y_{0}$ of order $2n-1$, namely,
\begin{equation}
u(y)-u(y_{0})=O\left((y-y_{0})^{2n}\right). \label{u(y) zero}
\end{equation}
Finally with (\ref{y(x)}) and (\ref{u(y) zero}), the 
function $u(x,t_{0})$ has an infinite number of singularities of the 
type $O\left((x-x_{0})^{\frac{2n}{2n+1}}\right)$, i.e., cusps. 
\end{proof}

\section{Numerical study of algebro-geometric solutions to the 
Camassa-Holm equation}
In this section we will numerically study concrete examples for the 
CH solutions (\ref{sol u CH}). As shown in the previous sections, 
real and bounded solutions are obtained on hyperelliptic M-curves, 
i.e., curves of the form
$$\mu^{2}=
  \prod_{i=1}^{2g+2}(\lambda-\lambda_{i}),
$$
where $g$ 
is the genus of the Riemann surface, and where we have for the branch 
points $\lambda_{i}\in \mathbb{R}$ the relations 
$\lambda_{i}\neq\lambda_{j}$  for $i\neq j$.

For the numerical evaluation  of the CH solutions (\ref{sol u CH}) we use the code presented in 
\cite{cam,lmp} for real hyperelliptic Riemann surfaces. The reader is 
referred to these publications for details. The basic idea is to 
introduce  a convenient homology basis 
on the related surfaces, see 
Fig.~\ref{cutsystem}.
\begin{figure}[htb!]
\begin{center}
   \includegraphics[width=0.7\textwidth]{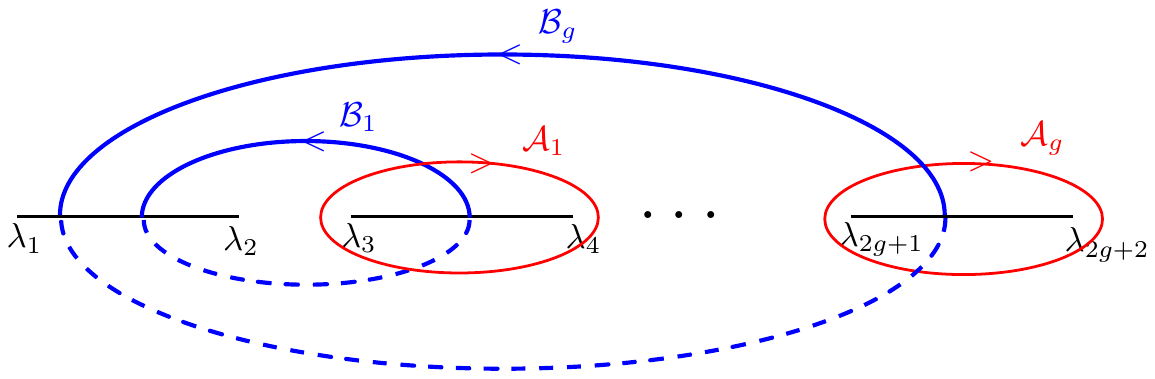}
\end{center}
 \caption{\textit{Homology basis on real hyperelliptic 
 M-curves, contours on sheet 1 are solid, contours on sheet 2 are dashed.}}
   \label{cutsystem}
\end{figure}
It is related to the basis used in the previous sections by the 
simple relation $\mathcal{A}\to -\mathcal{B}$, $\mathcal{B}\to 
\mathcal{A}$. This choice of the homology basis 
has the advantage that the limit in which branch points encircled by 
the same $\mathcal{A}$-cycle collide can be treated essentially numerically. 
Below we will consider examples where the distance of such a pair of branch points is 
of the order of machine precision ($10^{-14}$) and thus 
numerically zero. This limit is interesting since the $\mathcal{B}$-periods 
diverge, which implies that the corresponding theta functions reduce 
to elementary functions. The CH solutions (\ref{sol u CH}) reduce in 
this case to solitons or cuspons. Since we want to study also this 
limit numerically, we use the homology basis of Fig.~\ref{cutsystem}, 
and not the one of the previous sections.

The sheets are identified at the point $a$ by the sign of the root 
picked by Matlab. We denote a point in the first sheet with projection $\lambda$ into 
the complex plane by $\lambda^{(1)}$, and a point 
in the second sheet with the same projection by $\lambda^{(2)}$. 
The theta functions are in general approximated numerically by a 
truncated series as explained in \cite{cam} and \cite{BK}. A vector of 
holomorphic differentials for these  surfaces is given by 
$(1,\lambda,\ldots,\lambda^{g-1})^{t}\ud\lambda/\mu$.  The periods of 
the surface are computed as integrals between branch points of 
these differentials as detailed in \cite{lmp}. The Abel map of the 
point $a$ (and analogously for $b$) is computed in a similar way, see \cite{CK}, as 
the integral between $a$ and the branch point with minimal distance 
to $a$. It is well known (see for instance \cite{BBEIM}) that the 
Abel map between two branch points is a half period. 

To control the accuracy of the numerical solutions, we use 
essentially two approaches. First we check the theta identity 
(\ref{cor Fay1}), which is the underlying reason for the studied functions being 
solutions to CH. Since this identity is not built into the code, it 
provides a strong test. This check for various combinations of the 
points $a,b$ and $e$ ensures that the theta functions are computed 
with sufficient precision, and that the quantities $\alpha_{1}$ and 
$\alpha_{2}$ in (\ref{implicit equ}) are known with the wanted precision 
(we always use machine precision here). 
In addition, the smooth solutions are computed on Chebyshev 
collocation points (see, for instance, \cite{tref}) for $x$ and $t$. 
This can be used to approximate the 
computed solution via Chebyshev polynomials, a so-called 
spectral method having in 
practice exponential convergence for smooth functions. Since the 
derivatives of the Chebyshev polynomials can be expressed linearly in terms of Chebyshev 
polynomials, a derivative acts on the space of polynomials via a so 
called differentiation matrix. With these standard Chebyshev differentiation 
matrices (see \cite{tref}), the solution can be numerically 
differentiated. The computed derivatives allow to check with which 
numerical precision the partial differential equation (PDE) is satisfied by a numerical solution. With these two independent tests, we ensure that the 
shown solutions are correct to much better than plotting accuracy 
(the code reports a warning if the above tests are not satisfied to 
better than $10^{-6}$). We do not use the expansion in terms of 
Chebyshev polynomials for the cusped solutions since the convergence 
is slow for functions with cusps. 

We first consider smooth solutions $u$ (\ref{sol u CH}) in the case $\tau a 
=a$, $\tau b=b$. To obtain non-trivial solutions in the solitonic limit, 
we use  a vector $\mathbf{d}$ corresponding to the characteristics $\frac{1}{2}
\left[\begin{smallmatrix}
   1 & \ldots & 1  \\
   0 & \ldots & 0
\end{smallmatrix}\right]^{t}
$ in all examples. To plot a solution $u$ in dependence of $x$ and 
$t$, we compute it on a numerical grid for $y$ and $t$ to obtain 
$x(y,t)$ defined in (\ref{implicit equ}) and $u(y,t)$ given by (\ref{sol u CH}). These are then used to obtain a plot of 
$u(x,t)$ without having to solve the implicit relation  
(\ref{implicit equ}). In all examples we have $k=1$. The solutions for 
$\tau a=b$, e.g., for points $a$ and $b$ on the cuts encircled by the 
$\mathcal{A}$-periods in Fig.~\ref{cutsystem} look very similar and 
are therefore not shown here. 

Solutions on elliptic surfaces describe travelling waves and will not 
be discussed here. In genus 2 we obtain CH solutions of the form shown in 
Fig.~\ref{figch2}. The typical soliton collision known from the KdV 
equation is also present here, the unchanged shape of the solitons 
after the collision, but an asymptotic change of phase. 
\begin{figure}[htb!]
\begin{center}
\includegraphics[width=0.45\textwidth]{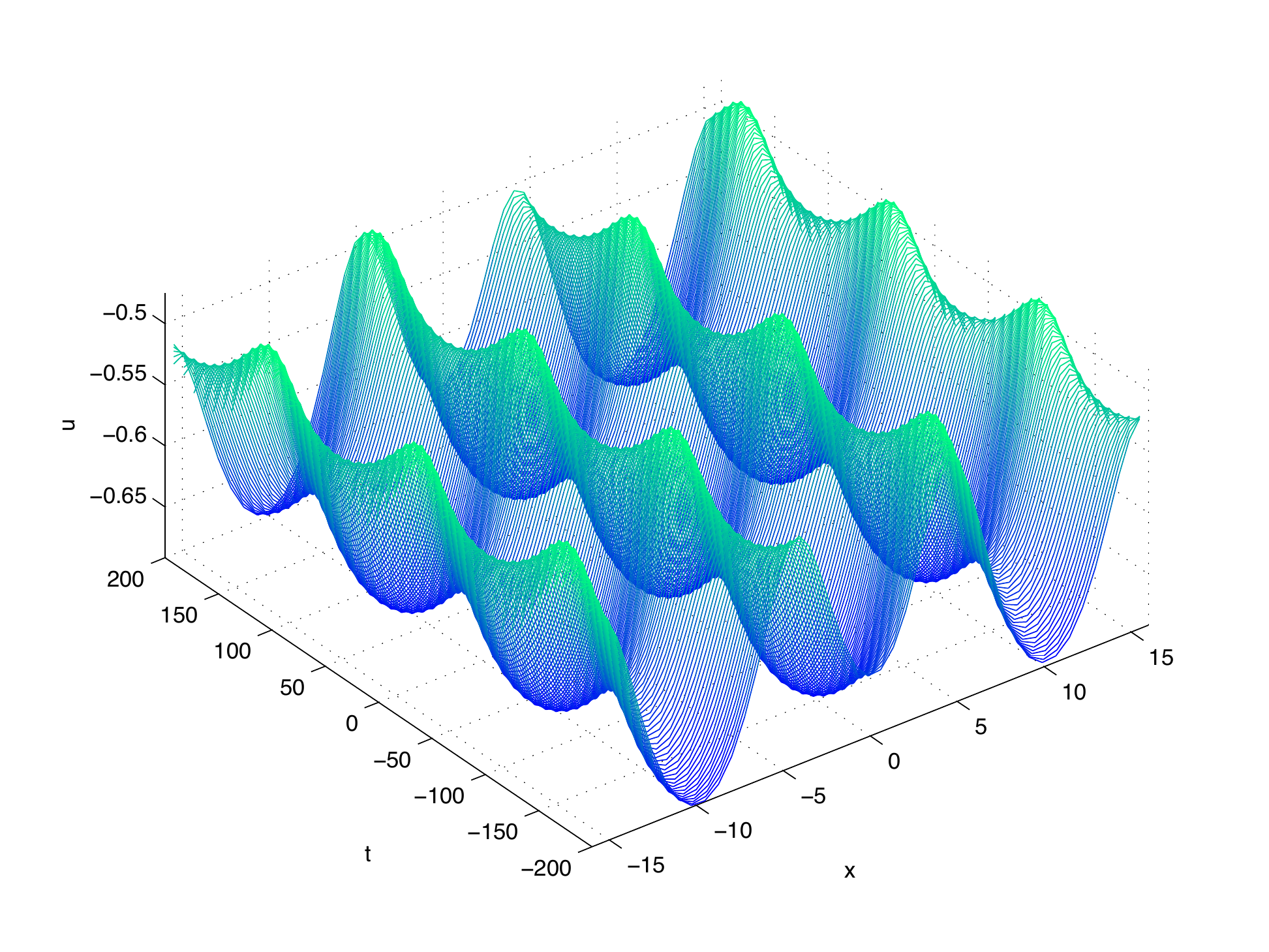}
\includegraphics[width=0.45\textwidth]{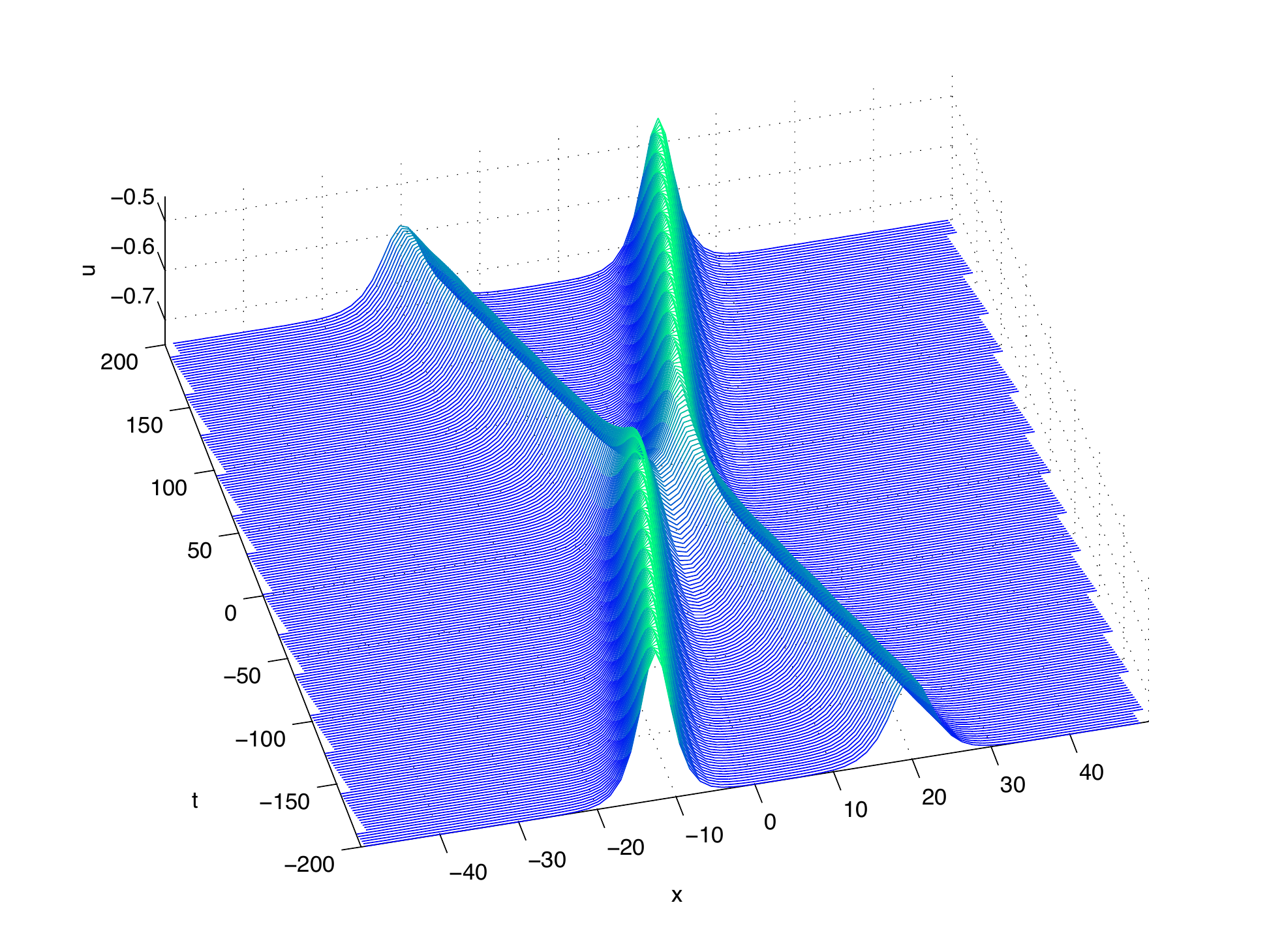}
\end{center}
 \caption{\textit{Solution (\ref{sol u CH}) to the CH equation  
 on a hyperelliptic curve of 
 genus 2 with branch points $-3,-2,0,\epsilon,2,2+\epsilon$ and 
 $a=(-4)^{(1)}$, $b=(-4)^{(2)}$ and $e=(-3,0)$ for $\epsilon=1$ on the left and 
 $\epsilon=10^{-14}$, the almost solitonic limit, on the right.}}
   \label{figch2}
\end{figure}

In genus 6 the CH solutions  have the form shown in 
Fig.~\ref{figch6}. In the solitonic limit one can recognize a 6 
soliton event. 
\begin{figure}[htb!]
\begin{center}
\includegraphics[width=0.45\textwidth]{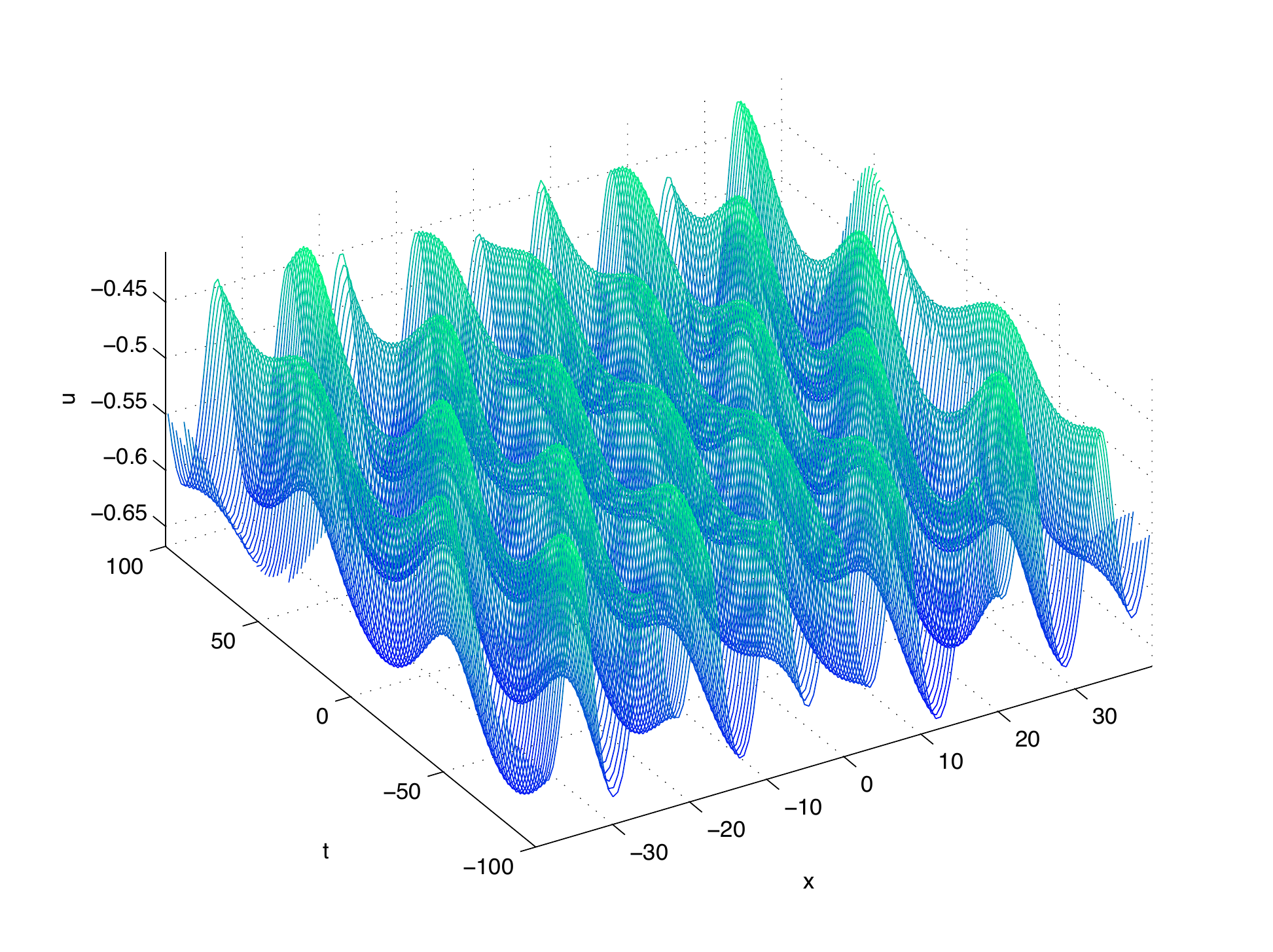}
\includegraphics[width=0.45\textwidth]{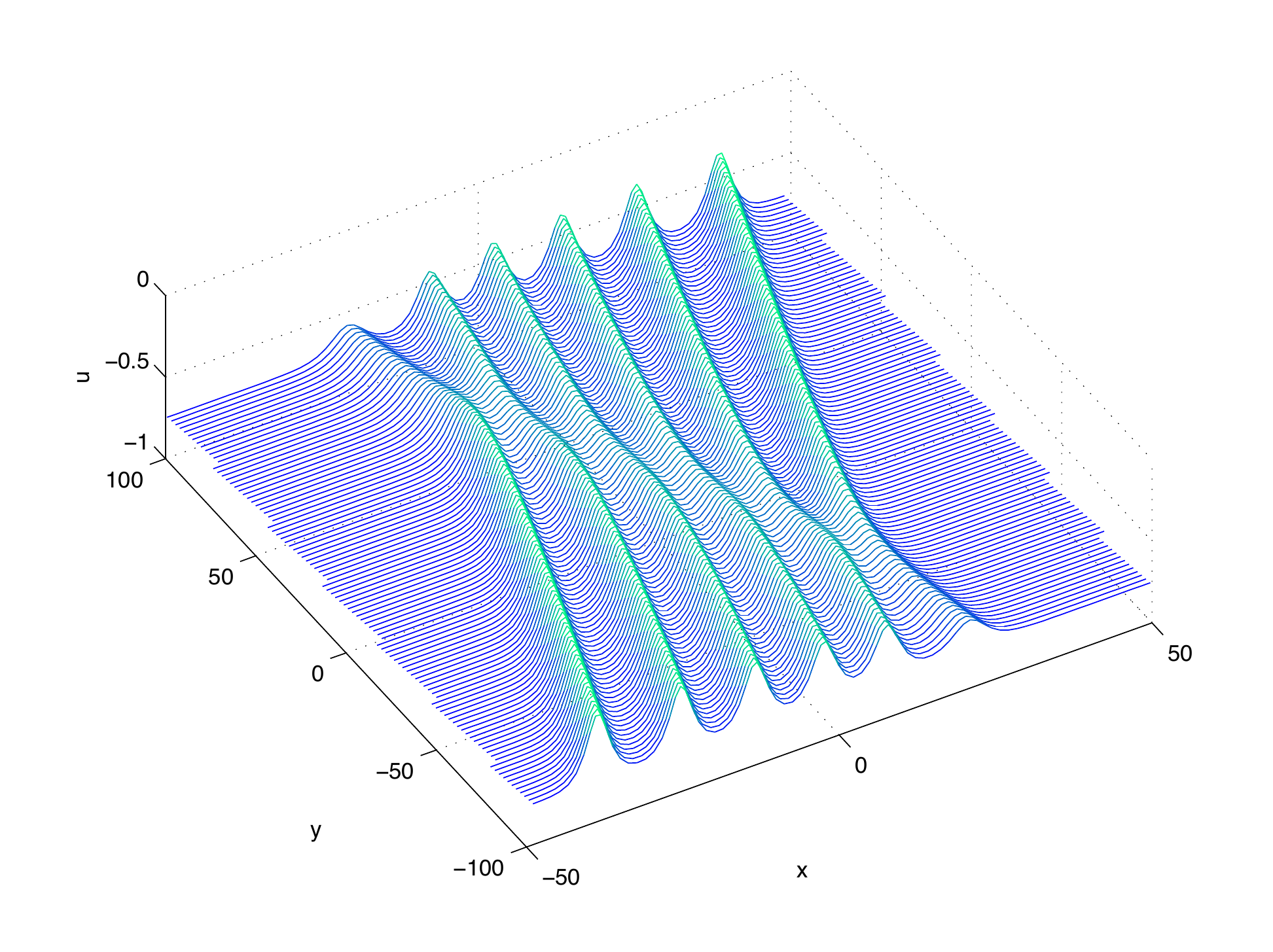}
\end{center}
 \caption{\textit{Solution (\ref{sol u CH}) to the CH equation  
 on a hyperelliptic curve of 
 genus 6 with branch points 
 $-7,-6,-5,-5+\epsilon,-3,-3+\epsilon,-1,-1+\epsilon,1,1+\epsilon,3,3+\epsilon,5,5+\epsilon$ and 
 $a=(-8)^{(1)}$, $b=(-8)^{(2)}$ and $e=(-7,0)$ for $\epsilon=1$ on the left and 
 $\epsilon=10^{-14}$, the almost solitonic limit, on the right.}}
   \label{figch6}
\end{figure}

The reality properties of the quantities entering the solution (\ref{sol u CH})
depend on the choice of the homology basis. For instance, in the 
homology basis of Fig.~\ref{cutsystem} and
for $a$ and $b$ stable under $\tau$, the Abel map 
$\mathbf{r}$ up to a vector proportional to $\mathrm{i}\pi$
and the vectors $\mathbf{V}_{b},\mathbf{V}_{e}$ are real, 
whereas these quantities are purely imaginary in the homology basis used in the previous sections.
Thus the easiest way 
to obtain cusped solutions is in this case to put $e=\lambda_{2}$ and 
to choose $\mathbf{d}$ corresponding to the characteristics $\frac{1}{2}
\left[\begin{smallmatrix}
   1 & \ldots & 1  \\
   1 & \ldots & 1
\end{smallmatrix}\right]^{t}
$. It can be easily checked that the theta functions 
$\Theta(\mathbf{Z}-\mathbf{d}\pm \mathbf{r}/2)$ cannot vanish since 
the argument is real, whereas the $\Theta(\mathbf{Z}-\mathbf{d})$ 
will have zeroes since the argument is complex. This implies that the 
derivative $u_{x}$ in (\ref{u_x}) diverges which corresponds to cups  
for the solution $u$. Peakons do not appear in such a limit of 
theta-functional solutions to CH and are thus not discussed here. To 
obtain them one would have to glue solutions in the solitonic limit 
on finite intervals to obtain a continuous solution that is piecewise 
$C^{1}$. 

We show the cusped solutions always in a 
comoving frame $x'=x+vt$ to allow  a better visualization of the 
solutions. In genus 2 we obtain cusped CH solutions of the form shown in 
Fig.~\ref{figch2cusp}, where also cusped solitons can be seen  in the 
degenerate situation. Obviously the collision between cuspons is 
analogous to soliton collisions.
\begin{figure}[htb!]
\begin{center}
\includegraphics[width=0.45\textwidth]{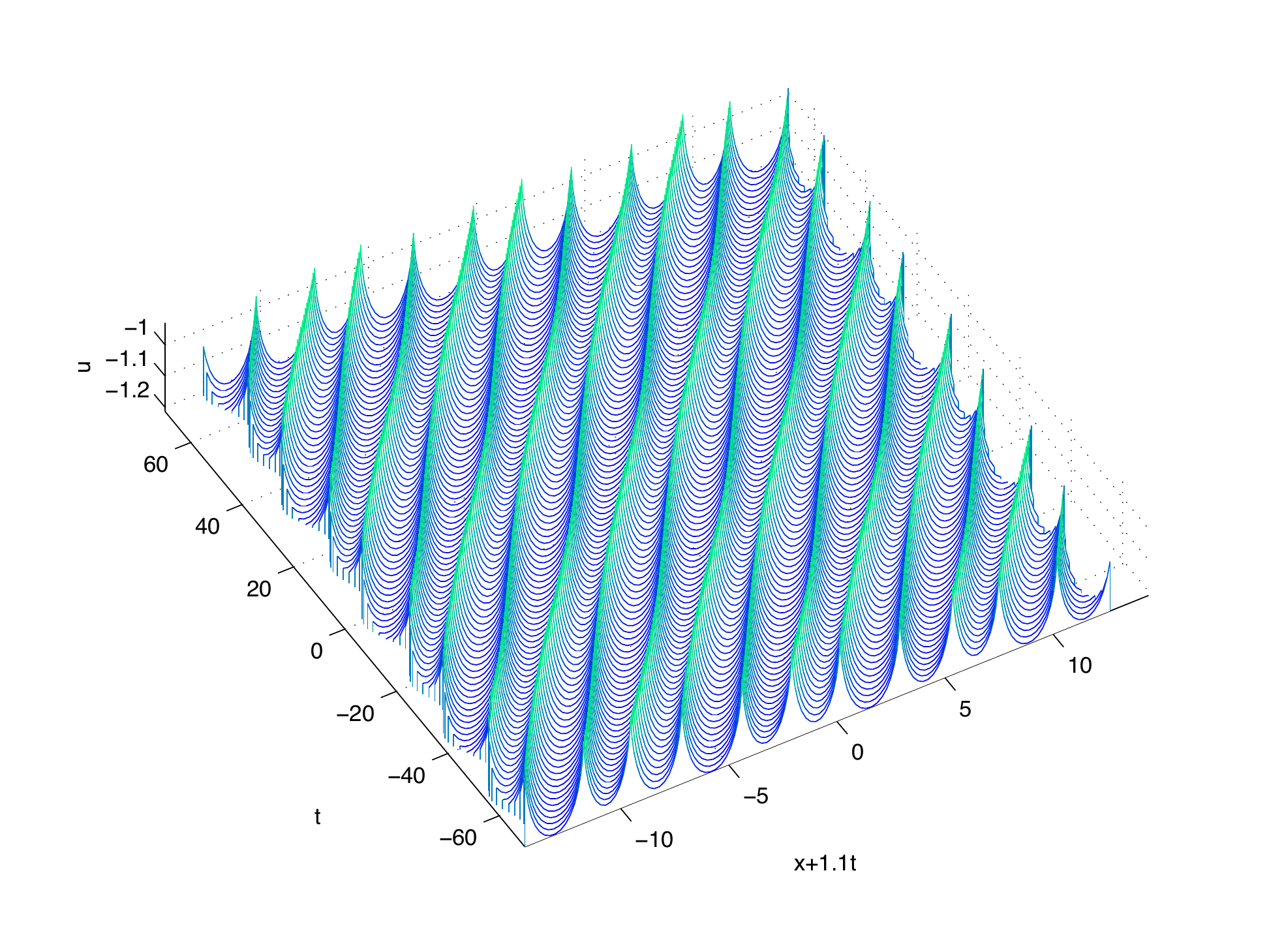}
\includegraphics[width=0.45\textwidth]{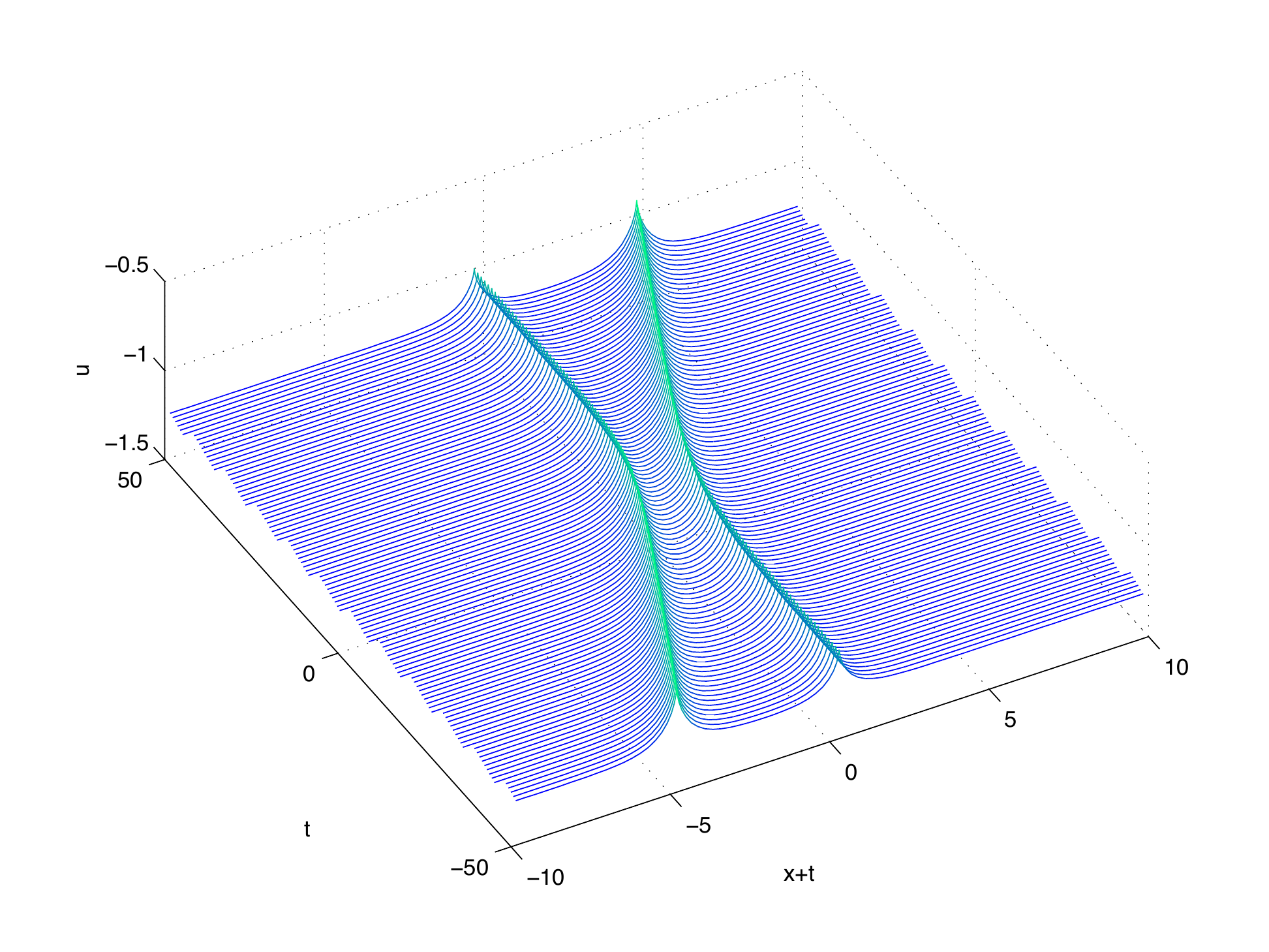}
\end{center}
 \caption{\textit{Cusped solution (\ref{sol u CH}) to the CH equation  
 on a hyperelliptic curve of 
 genus 2 with branch points $-3,-2,0,\epsilon,2,2+\epsilon$ and 
 $a=(-4)^{(1)}$, $b=(-4)^{(2)}$ and $e=(-2,0)$ for $\epsilon=1$ on the left and 
 $\epsilon=10^{-14}$, the almost solitonic limit, on the right.}}
   \label{figch2cusp}
\end{figure}

In genus 6 the CH solutions  have the form shown in 
Fig.~\ref{figch6cusp}. In the solitonic limit one can recognize a 6-cuspon. 
\begin{figure}[htb!]
\begin{center}
\includegraphics[width=0.45\textwidth]{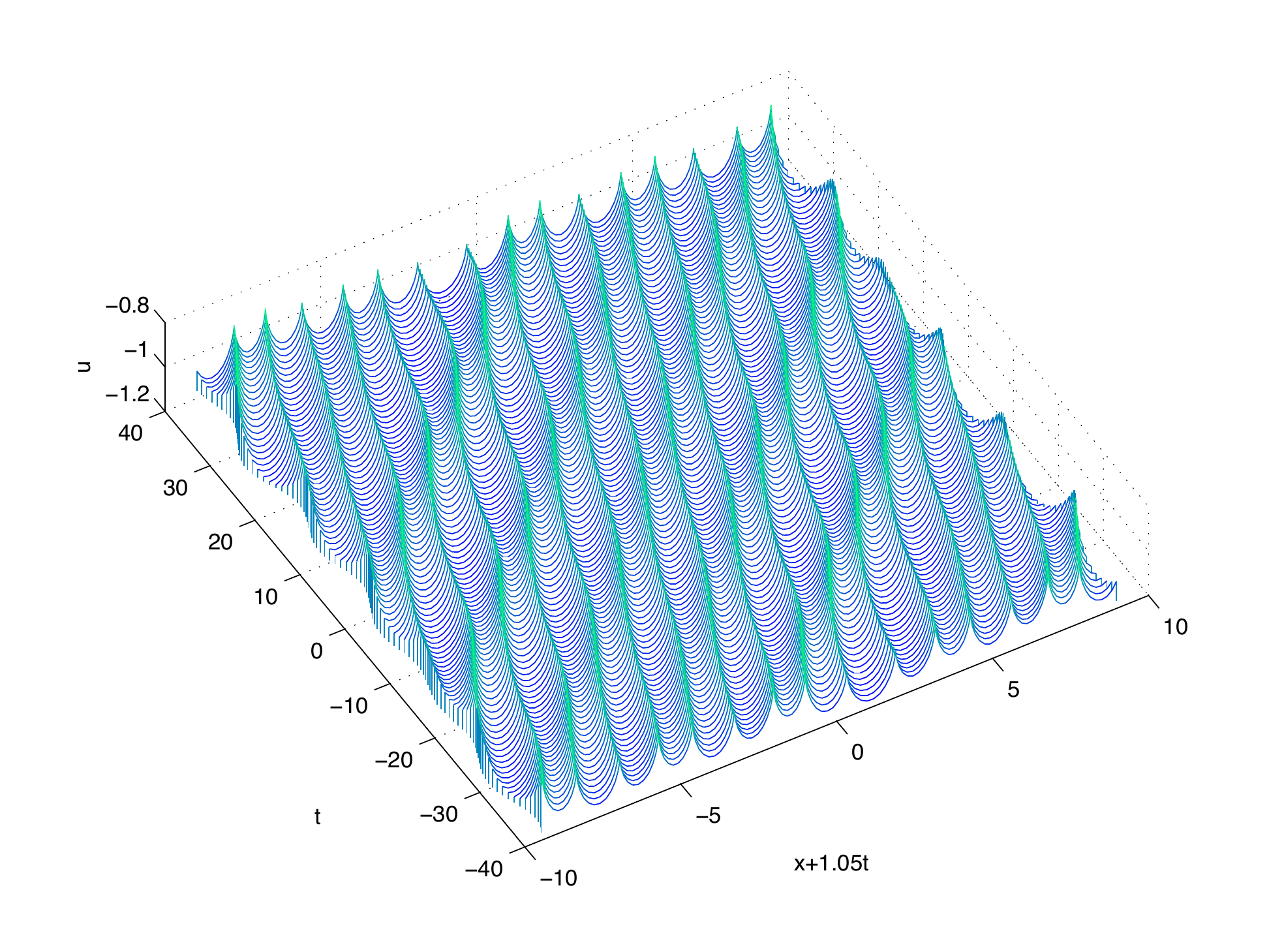}
\includegraphics[width=0.45\textwidth]{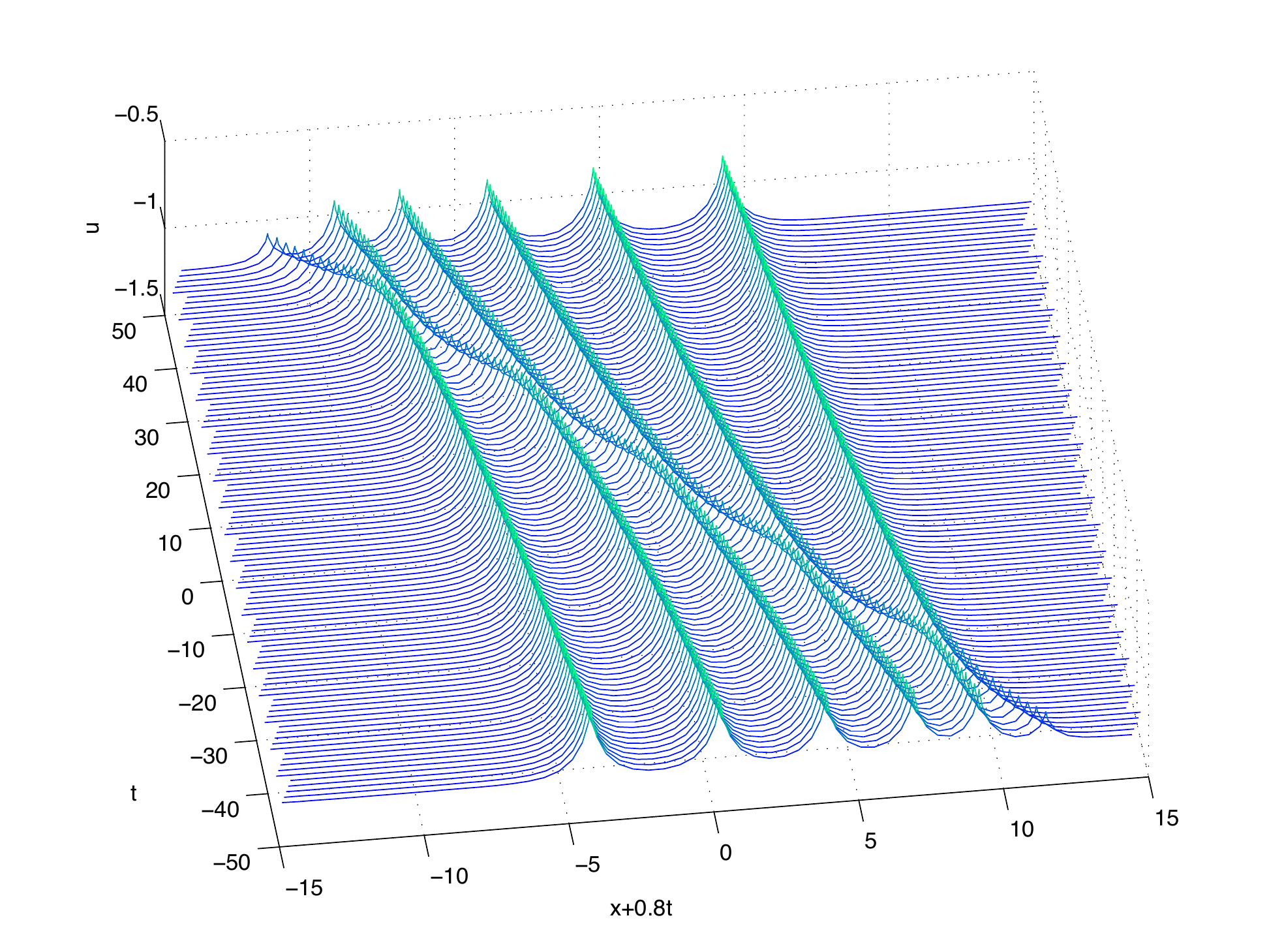}
\end{center}
 \caption{\textit{Solution (\ref{sol u CH}) to the CH equation  
 on a hyperelliptic curve of 
 genus 6 with branch points 
 $-7,-6,-5,-5+\epsilon,-3,-3+\epsilon,-1,-1+\epsilon,1,1+\epsilon,3,3+\epsilon,5,5+\epsilon$ and 
 $a=(-8)^{(1)}$, $b=(-8)^{(2)}$ and $e=(-6,0)$ for $\epsilon=1$ on the left and 
 $\epsilon=10^{-14}$, the almost solitonic limit, on the right.}}
   \label{figch6cusp}
\end{figure}

\section{Conclusion}
In this paper we have shown at the example of the CH equation
that Mumford's program to construct 
algebro-geometric solutions to integrable PDEs can be also applied to 
non-local (here in $x$) equations. For the studied case the solutions in terms of 
multi-dimensional theta functions do not depend directly on the physical 
coordinates $x $ and $t$, but via an implicit function. One 
consequence of this non-locality is the existence of non-smooth 
solitons. A numerical study of smooth and non-smooth solutions was 
presented.

A further 
example in this context would be 
the equation from the Dym-hierarchy for which theta-functional 
solutions were studied in \cite{AFY1}, which will be treated 
elsewhere with Mumford's approach. It is an interesting question 
whether the 2+1 dimensional generalization of the CH equation 
\cite{falqui,zang}, for which algebro-geometric solutions are so far 
unknown,  can be also treated with these 
methods.


\begin{thebibliography}{99}

\bibitem{ACFYHM2} M.S. Alber, R. Camassa, Yu.N. Fedorov, D.D. Holm, J.E.  Marsden, \textit{The complex geometry of weak piecewise smooth solutions
of integrable nonlinear PDE's of shallow water and Dym type}, Comm.
Math. Phys. \textbf{221}, 197--227 (2001).
\bibitem{AFY1} M.S. Alber,  Yu.N. Fedorov \textit{Wave solutions of evolution equations
and Hamiltonian flows on nonlinear subvarieties of generalized Jacobians},
J. Phys. A \textbf{33}, 8409--8425 (2000).
\bibitem{AFY2} M.S. Alber,  Yu.N. Fedorov,  \textit{Algebraic geometrical solutions for
certain evolution equations and Hamiltonian flows on nonlinear subvarieties
of generalized Jacobians}, Inverse Problems \textbf{17}, 1017--1042  (2001).





\bibitem{BSS1} W. Beals, D. Sattinger, J. Szmigielski,   \textit{Acoustic scattering and the extended Korteweg de Vries
hierarchy}, Adv. Math. \textbf{140}, 190--206 (1998).
\bibitem{BSS2} W. Beals, D. Sattinger, J. Szmigielski, \textit{ Multi-peakons and a theorem of Stietjes}, Inverse Problems
\textbf{15}, L1--L4 (1999).
\bibitem{BSS3} W. Beals, D. Sattinger, J. Szmigielski, \textit{Multipeakons and the classical moment}, Adv. in Math.
\textbf{154}, no. 2, 229--257 (2000).

 \bibitem{BBEIM} E. Belokolos, A. Bobenko, V. Enolskii, A. Its, V. Matveev, \textit{Algebro-geometric approach to nonlinear integrable equations}, Springer Series in nonlinear dynamics (1994).
 \bibitem{BK} A.I. Bobenko,  C. Klein, (ed.), \textit{Computational Approach to Riemann 
Surfaces}, Lect. Notes Math. \textbf{2013} (2011).

\bibitem{Cal} F. Calogero,  \textit{An integrable Hamiltonian system}, Phys. Lett. A \textbf{201}, 306--310 (1995).
\bibitem{CF} F. Calogero, J.-P.  Fran\c coise,  \textit{Solvable quantum version of an integrable Hamiltonian system}, J.
Math. Phys. \textbf{37}, (6) 2863--2871 (1996).
\bibitem{Cam} R. Camassa,  \textit{ Characteristic variables for a completely integrable shallow water equation}, In: Boiti,
M. et al. (eds.) Nonlinearity, Integrability and All That: Twenty Years After NEEDS'79. Singapore: World
Scientific (2000).
\bibitem{CH} R. Camassa,  D.D. Holm, \textit{An integrable shallow water equation with peaked solitons}, Phys. Rev. Lett.
\textbf{71}, (11) 1661--1664 (1993).
\bibitem{CHH} R. Camassa, D.D. Holm, J.M. Hyman, \textit{A new integrable shallow water equation}, Adv.
Appl. Mech. \textbf{31}, 1--33 (1994).
\bibitem{zang} M. Chen, S. Liu, Y. Zhang, \textit{A two-component 
generalization of the Camassa-Holm
equation and its solutions}, Lett. Math. Phys. \textbf{75}, 1--15 (2005).
\bibitem{CG} A. Clebsch, P. Gordan, \textit{Theorie der Abelschen Funktionen}, Teubner, Leipzig (1866).
\bibitem{DN} B.A. Dubrovin, S. Natanzon, {\it Real theta function solutions of the Kadomtsev-Petviashvili equation}, Math. USSR Irvestiya  \textbf{32}:2, 269--288 (1989).

\bibitem{falqui} G. Falqui, \textit{On a Camassa-Holm type equation with two dependent variables}, J. Phys.
A: Math. Gen. \textbf{39} (2006).
\bibitem{Fay} J. Fay, {\it Theta functions on Riemann surfaces}, Lecture Notes in Mathematics \textbf{352} (1973).

\bibitem{FF} A.S. Fokas, B. Fuchssteiner, \textit{Symplectic structures, their B\"acklund transformations and hereditary
symmetries}, Physica D \textbf{4}, 47--66 (1981/82).
\bibitem{cam}     J.~Frauendiener, C.~Klein, \textit{Hyperelliptic theta functions 	and spectral methods}, J. Comp. Appl. Math. (2004).    
\bibitem{lmp} J.~Frauendiener, C.~Klein, \textit{Hyperelliptic theta functions	and spectral methods: KdV and KP solutions},		Lett. Math. Phys., Vol. \textbf{76}, 249--267 (2006). 
\bibitem{GH1}F.~Gesztesy and H. Holden, \textit{Real-valued 
algebro-geometric solutions of the Camassa-Holm hierarchy}, Phil. 
Trans. R. Soc. A, \textbf{366}, 1025-1054 (2008).
\bibitem{GH2}F.~Gesztesy and H. Holden, \textit{Algebro-Geometric 
Solutions of the Camassa-Holm hierarchy}, Rev. Mat. Iberoamericana 
\textbf{19}, 73-142 (2003).
\bibitem{GH3}F.~Gesztesy and H. Holden, \textit{Soliton Equations and 
Their Algebro-Geometric Solutions,
Volume I: (1+1)-Dimensional Continuous Models}, Cambridge studies in 
advanced mathematics \textbf{79}
(Cambridge University Press, Cambridge, 2003).
\bibitem{Gross} B.H. Gross, J. Harris, \textit{Real algebraic curves}, Ann. sci. Ecole Norm. Sup. (4) \textbf{14},
157--182  (1981).
\bibitem{Harnack} A. Harnack, \textit{Ueber die Vieltheiligkeit der ebenen algebraischen Curven}, Math. Ann. \textbf{10}, 189--199 (1876).
\bibitem{Kalla} C. Kalla, \textit{New degeneration of Fay's identity and its application to integrable systems}, preprint arXiv:1104.2568v1 (2011).
\bibitem{CK} C.~Kalla and C.~Klein, \emph{On the numerical evaluation of algebro-geometric solutions to integrable equations} (2011)  arXiv:1107.2108.
\bibitem{KC} H.P. McKean,  A. Constantin,  \textit{A shallow water equation on the circle}, Comm. Pure Appl. Math.
Vol LII, 949--982 (1999).
\bibitem{Mum} D. Mumford, {\it Tata Lectures on Theta. I and II.}, Progress in Mathematics \textbf{28} and \textbf{43},
respectively. Birkh\"auser Boston, Inc., Boston, MA (1983 and 1984).
\bibitem{tref} L.N. Trefethen, \textit{Spectral Methods in  Matlab}, SIAM, Philadelphia, PA (2000).
\bibitem{Vin} V. Vinnikov, {\it Self-adjoint determinantal representations of real plane curves}, Math. Ann. \textbf{296},
453--479 (1993).


\end{thebibliography}
\end{document}